\documentclass[floatfix,nofootinbib,superscriptaddress,twocolumn,pre,showpacs,showkeys,preprintnumbers]{revtex4-2}

\usepackage{latexsym,amsmath,amsfonts,amsthm,amssymb,bbm}
\usepackage{epsfig,graphics,color,calc,graphicx,pict2e}
\usepackage{comment} 
\usepackage{multirow}
\usepackage{booktabs} 
\usepackage{footmisc} 
\usepackage[T1]{fontenc}
\usepackage[utf8]{inputenc}
\usepackage{float} 

\usepackage{mathrsfs}
\usepackage{tikz}
\usetikzlibrary{shapes,arrows}
\usetikzlibrary{intersections,matrix,positioning}

\usepackage[printonlyused,withpage]{acronym}
\usepackage[acronym]{glossaries}

\usepackage{subfigure} 
\newtheorem{theorem}{Theorem}

\newtheorem{lemma}{Lemma}
\newtheorem{remark}{Remark}

\newlength{\pecettawidth}
\newcommand{\bbZ}{\mathbb{Z}}
\newcommand{\xor}{\oplus}

\newacronym{pca}{PCA}{Probabilistic Cellular Automata}
\newacronym{ca}{CA}{Cellular Automata}
\newacronym{eca}{ECA}{Elementary Cellular Automata}
\newacronym{deca}{DECA}{Diploid Elementary Cellular Automata}
\newacronym{ndeca}{NDECA}{Null Diploid Elementary Cellular Automata}

\begin{document}
\title{\normalsize\Large\bfseries  Inefficiency of  the block approximation in diploid  Probabilistic Cellular Automata}


\author{Emilio N.M. Cirillo\thanks{emilio.cirillo@uniroma1.it}}
\affiliation{Dipartimento di Scienze di Base e Applicate per l'Ingegneria, 
             Sapienza Universit\`a di Roma, 
             via A.\ Scarpa 16, I--00161, Roma, Italy.}
\author{Joram L. Vliem\thanks{J.l.Vliem@students.uu.nl}}
\affiliation{Institute of Mathematics,
University of Utrecht, Budapestlaan 6, 3584 CD Utrecht, The~Netherlands.}
\affiliation{Institute for Theoretical Physics Utrecht University, Princetonplein 5, 3584 CC Utrecht, The Netherlands}
\author{Dirk Schuricht\thanks{D.Schuricht@uu.nl}}
\affiliation{Institute for Theoretical Physics Utrecht University, Princetonplein 5, 3584 CC Utrecht, The Netherlands}
\author{Cristian Spitoni\thanks{C.Spitoni@uu.nl}}
\affiliation{Institute of Mathematics,
University of Utrecht, Budapestlaan 6, 3584 CD Utrecht, The~Netherlands.}


\begin{abstract}
We study a probabilistic cellular automaton obtained as a mixture of the additive elementary rules 60 and 102. We prove that, for any finite periodic lattice and for mixing parameter $\lambda=1/2$, the system almost surely reaches the absorbing all-zero configuration in finitely many steps. In addition, Monte Carlo simulations indicate as well the presence of a zero-density stationary state in a finite interval around $\lambda=1/2$.
Despite this absorbing behavior, both mean-field and block approximation schemes predict a stationary state with non-zero density. This failure, traced to the additive and mirror symmetries of the deterministic components, highlights a fundamental limitation of finite-block approximation in capturing the global dynamics of probabilistic cellular automata. 
\end{abstract}

\maketitle









\section{Introduction}
\label{s:intro} 
\gls{pca} are a stochastic 
generalization of deterministic \gls{ca}. They are defined as 
discrete-time Markov chains whose state space is the product of a finite single-site state space.
The cell variables 
are simultaneously updated according to 
a single-cell
probability distribution which depends on the 
neighboring cells configuration at the previous time.
Despite their simplicity,
\gls{pca} exhibit an ample variety of dynamical behaviors and are thus
considered powerful and useful modelling tools.
We refer the interested reader to 
\cite{LN2016} for a comprehensive introduction to the subject and to 
\cite{VRD2023,D2020} for recent applications.

The focus of this paper is on the stationary behavior 
of a particular 
class of \gls{pca} which are a probabilistic mixture of 
\gls{eca} 
that is to say, 
one-dimensional \gls{ca} in which 
the state of a cell, either zero or one,
is chosen
depending on its previous state 
and on that 
of its two neighboring 
cells 
\cite{W83,W84}. 
Binary
probabilistic 
mixtures are called \emph{diploid} and have been widely studied 
in the literature 
\cite{Fautomata2017,BMM2013,T2004,MM2014,Me2011,Dh83}
not only for their intrinsic
interest, but also because some of them encode the evolution 
rule of very well known paradigmatic \gls{pca} models.

In this work we study the diploid PCA obtained as a probabilistic mixture (with mixing parameter $\lambda$) of the
additive ECA rules 60 and 102  which are left-right reflections of one another (see Table~\ref{tab:rules_60_102}).
Although both rules are linear over $\mathbb{Z}_2$ and individually generate
Sierpiński-type patterns (see \cite{wolfram2003}), their mixture exhibits an interesting absorbing-state
phenomenon.  
Monte Carlo simulations indicate the presence of an \emph{inactive} or
\emph{zero-density}  stationary state in an interval of $\lambda$ values around
$\lambda=1/2$, separated by two symmetric transition points from a
non-zero density regime.  
 
This absorbing behavior stands in sharp contrast to what is predicted by common approximation methods.  
Both the Mean-Field (MF) and the Block Approximation (BA) with block sizes up to
$M=13$ produce positive stationary densities for all $\lambda\in(0,1)$, and
therefore fail to detect the zero-density stationary state suggested by the simulations and
rigorously proven at the symmetric point.  
A detailed analysis of the BA shows that the block-approximation density
$\rho_M(\lambda)$ decays exponentially with $M$, and an extrapolation to
infinite block size yields values compatible with zero near
$\lambda=1/2$.  
However, the effective block size required to reproduce this behavior is of
order $M\approx 38$, far beyond computationally feasible block sizes.  

This work has two main objectives. First, as discussed above, we show that the MF and BA schemes fail to capture the transition between the zero-density and positive-density states. Second, we rigorously establish the existence of the zero-density state at $\lambda = 1/2$. To this end, we develop a rigorous argument based on estimating the probability that the dynamics maps an arbitrary configuration to the all-zero configuration within a finite number of time steps. Our central result is that this probability is strictly positive. In addition, Monte Carlo simulations indicate that the zero-density phase extends to a finite interval around $\lambda=1/2$. 

The remainder of the paper is organized as follows.
Section~\ref{s:modello}  recalls the definition of diploid PCA and reviews
the BA scheme.  
Section~\ref{s:diploid} introduces the diploid mixture of rules 60 and
102, describes its additive properties, and presents numerical simulations.
Section~\ref{s:rigorous} contains the rigorous proof of absorption for 
$\lambda=1/2$, i.e., the all-zero configuration is almost surely reached in finite time.  
Section~\ref{s:failure} analyses in detail the behavior of MF and BA, including an exponential scaling study of the block-size dependence.  
Finally, Section~\ref{s:future} summarizes the results and outlines future directions.

\section{Mixture of ECA and block approximation}
\label{s:modello} 
We introduce in brief the probabilistic mixtures of 
\gls{eca} and the BA approach. 
We refer to \cite{CLS2024} for a more detailed discussion
and for several examples of applications. 

Call 
$Q=\{0,1\}$ 
the
\emph{set of states} 
and 
$\mathbb{Z}_n=\mathbb{Z}/n\mathbb{Z}=\{0,1,\dots,n-1\}$ 
the
$n$ cell
annulus,
with $n$ an integer.
The \emph{configuration space}
is 
$X_n=Q^{\mathbb{Z}_n}$.
For $x\in X_n$, 
$\Delta\subset\mathbb{Z}_n$,
the restriction of $x$ to $\Delta$
is $x_\Delta$.
We abuse notation by writing 
$x_i$ for $x_{\{i\}}$ and call it 
the \emph{value} of 
the cell $i$ or the \emph{occupation number} of the cell $i$. 

Given the map $F:X_n\to X_n$,
the \gls{ca} associated with $F$ 
is the 
collection of the sequences of configurations 
$\zeta^t$, with $t\in\mathbb{N}_0$ such that 
$\zeta^t=F(\zeta^{t-1})$.
The sequence $\zeta^t$ such that 
$\zeta^0=x_0\in X_n$ 
is the \emph{trajectory} of the \gls{ca} with \emph{initial condition} $x_0$. 

In 
\gls{eca}
cell states 
are updated according to the state 
of the 
\emph{neighborhood} 
$I_i=\{i-1,i,i+1\}$
of the cell
$i\in\mathbb{Z}_n$.
Note that, due to the periodicity of $\mathbb{Z}_n$, 
the neighborhood of the origin is 
$I_0=\{n-1,0,1\}$.
Thus, 
given a function 
$f:Q^3\to Q$, called
\emph{local rule}, 
the associated \gls{eca} is the \gls{ca} defined by the 
map
$F:X_n\to X_n$ such that
$(F(x))_i = f(x_{I_i}),$ for any $i\in\mathbb{Z}_n$ and $x\in X_n$.

The 
256 possible \gls{eca} 
\cite{wolfram2003}
are identified 
by the integer 
number 
\begin{align}
\label{fin000}
W
=&
\phantom{+}
f(1,1,1)\cdot2^7
+
f(1,1,0)\cdot2^6
+
f(1,0,1)\cdot2^5
\notag\\
&
+
f(1,0,0)\cdot2^4
+
f(0,1,1)\cdot2^3
+
f(0,1,0)\cdot2^2
\notag\\
&
+
f(0,0,1)\cdot2^1
+
f(0,0,0)\cdot2^0
\end{align}
belonging to $\{0,\dots,255\}$.
The coefficients of the above expansion provide
the binary representation of the number $W$ and are 
sometimes used to denote the \gls{eca}. 
Note that for the given  \gls{eca} $W$, the 
\textit{conjugate under left-right reflection} 
rule is obtained by exchanging $f(1,1,0)$ with $f(0,1,1)$ 
and $f(1,0,0)$ with $f(0,0,1)$. 

\gls{pca} are defined as Markov chains 
$\zeta^t$, with $t\in\mathbb{N}_0$, on $X_n$ 
with transition matrix 
from state $x$ to $y$ 
\begin{equation}
\label{mod000a}
p(x,y)
=
\prod_{i\in\mathbb{Z}_n}
p_i(y_i|x),
\end{equation}
with the \emph{one-site updating probability}
$p_i(\cdot|x)$ being 
a probability distribution on $Q$ parameterized 
by the configuration $x$. 
Thus, at each time all the cells are updated 
simultaneously and independently. 
We denote with $\mathbb{P}_x$ the probability associated with the chain
started at $x$ and, thus,
$\mathbb{P}_x(\zeta^t=y)$ 
(resp.~$\mathbb{P}_x(\zeta^t\in Y)$)
is the probability 
that the chain started at $x$ is at time $t$ 
in the configuration $y$ 
(resp.~in the set of 
configurations $Y\subset X_n$).

\gls{pca} can be defined as probabilistic
mixtures of \gls{eca} by choosing 
\begin{equation}
\label{mod000}
p_i(y_i|x)
=
y_i\phi(x_{I_i})
+
(1-y_i)[1-\phi(x_{I_i})],
\end{equation}
where
$\phi:Q^3\to[0,1]$, which can 
be interpreted as the probability to set the cell to one,
is defined as 
\begin{equation}
\label{mod010}
\phi=\sum_{r=1}^m \xi_r f_r,
\end{equation} 
with $m$ a positive integer,
$\xi_r\in[0,1]$ such that $\xi_1+\cdots+\xi_m=1$
and $f_1,\dots,f_m$ are $m$ local rules.
The following interpretation is in order:
at each time and for each
cell in $\mathbb{Z}_n$ one chooses at random one integer $r$ in $\{1,\dots,m\}$ 
with probabilities
$\xi_1,\dots,\xi_m$ and 
performs the updating with the rule $f_r$ based on the 
neighborhood configuration at time $t-1$.
Indeed, with this algorithm the
probability to set the cell to $1$ at time $t$ is
the sum of the $\xi_r$ such that $f_r$, 
computed in the neighborhood configuration at time $t-1$,
is one.

The case $m=2$ has been widely studied in the literature
\cite{paul2026,mendoncca2011,fates2008,CNS2021,Fautomata2017,CLS2024}: 
the probabilistic 
mixtures of \gls{eca} are called \emph{diploid} and 
\eqref{mod010} simplifies to 
\begin{equation}
\label{mod010bis}
\phi=(1-\lambda)f_1+\lambda f_2,
\end{equation} 
with $\lambda\in[0,1]$.

We will study the 
possibility for \gls{pca} to exhibit 
stationary behaviors. For a generic probability measure $\mu$ on $X_n$ we define the density as the
spatially averaged occupation number
\begin{equation}
\label{e:densityG}
\delta_\mu := \frac{1}{n}\sum_{i\in\bbZ_n}\mu(x_i=1). 
\end{equation}
If $\mu$ is translation invariant, then $\delta_\mu=\mu(x_0=1)$. We will call \emph{density} the quantity $\delta_\mu$. 
Moreover, we shall also use 
\begin{equation}
\label{e:densityMC}
\delta_x=\frac{1}{T}
\sum_{t=1}^T\frac{1}{n}\sum_{i\in \mathbb{Z}_n } \zeta_i^t
,
\end{equation}
as its empirical estimate in Monte Carlo simulations.
The estimate can be improved by starting to average from a suitable 
thermalization time and sampling the dynamics not at each time but
at times separated by a suitable decorrelation time.
We want to stress that in the numerical simulations we will always initialize the PCA from an i.i.d.\
Bernoulli product measure on $X_n$. Since the update rule is the same at
every site, the Markov evolution commutes with spatial translations; hence the
distribution of $\zeta^t$ is translation invariant for every $t\ge0$.
In particular, $\mathbb E[\zeta_i^t]$ does not depend on $i$, and the expected
density can be identified with $\mathbb E[\zeta_0^t]$.

Block approximations of the stationary measure have been shown to be reliable in several studies, see, 
e.g., 
\cite{CLS2024,FF15} and references therein. 
We shortly recall the main ideas, slightly modifying the notation 
with respect to the one adopted in \cite{CLS2024}:
let a \textit{block} be a sequence of contiguous cells and denote 
$B_{i,m}$ the $m$-cell block 
$\{i,i+1,\dots,i+m-1\}$, with $i\in\mathbb{Z}_n$ and 
$1\le m\le n-i$; recall $\mathbb{Z}_n$ is an annulus.
Note that the neighborhood $I_i$ defined above is nothing 
but the block $B_{i-1,3}$.
Thus, given $z$ a configuration on $B_{i,m}$ we write 
\begin{align}
\label{eq:it010}
P_x(\zeta_{B_{i,m}}^t & =z) 
\notag\\
=&
\sum_{y\in Q^{B_{i-1,m+2}}} 
 \prod_{k\in B_{i,m}}p_k(z_k|y)  
                     P_x(\zeta^{t-1}_{B_{i-1,m+2}}=y) 
.
\end{align}
Now, let $\varrho$ be the stationary 
distribution and abuse the notation by writing
\begin{equation}
\label{eq:it012}
\varrho(z)
=
\sum_{y\in X_n:\, y_{B_{i,m}}=z} \varrho(y)
\end{equation}
for any $z\in Q^{B_{i,m}}$.
Thus, by 
exploiting periodic boundary conditions and translation invariance,
from \eqref{eq:it010}
we get
\begin{equation}
\label{eq:it014}
\varrho(z) 
=
\sum_{y\in Q^{B_{0,m+2}}} 
 \prod_{k\in B_{1,m}} p_k(z_k|y)  \varrho(y) 
\end{equation}
for any $z\in Q^{B_{1,m}}$.
We remark that the left-hand side of the equation above contains the $2^m$ 
values of the stationary probability of the configurations on the 
$m$-block $B_{1,m}$.
Instead, the right-hand side is written 
in terms of the $2^{m+2}$ values of the stationary probability 
of the configurations on the $(m+2)$-block $B_{0,m+2}$.

The $M$-cell BA of the stationary measure 
consists in 
setting $M$ as the maximal block size and 
rewriting
\eqref{eq:it014}
for the measure on the $M$-cell block
with
the $M+2$-cell measures appearing on the right-hand side
approximated 
by combinations of probabilities associated with smaller blocks.

A very smart approximation is provided 
by the local structure scheme based on the so-called Bayes' rule 
for BA \cite{GVK87,GV87,FF15}, 
namely, 
\begin{equation}
\label{eq:it030}
\varrho(x_1x_2)\approx\varrho(x_1)\varrho(x_2),
\end{equation}
and, for $m\ge3$, 
\begin{equation}
\label{eq:it040}
\varrho(x_1\dots x_m)
\approx
\frac{\varrho(x_1\dots x_{m-1})\varrho(x_2\dots x_m)}
      {\varrho(x_2\dots x_{m-1})}
.
\end{equation}
In the case in which
the denominator vanishes, we assume $\varrho(x_1\dots x_m)=0$.

This approximation ensures that the consistency property 
of block measures is satisfied, that is to say, 
the probability of smaller blocks can be computed
by considering the marginal of larger blocks measure,
by integrating out either the left-hand or the right-hand cell variable, 
more precisely, 
\begin{equation*}
\varrho(x_1\dots x_m)
=
\sum_{y\in Q} \varrho(yx_1\dots x_m)
=
\sum_{y\in Q} \varrho(x_1\dots x_my)
\,,
\end{equation*}
with $x_1,\dots,x_m\in Q$.
This allows to solve
equations~\eqref{eq:it014}--\eqref{eq:it040} using
the recursive scheme proposed in \cite{K74} within the framework
of the cluster variation method and called the natural iteration method.
The cluster variation method has been widely applied in
the context of the statistical mechanics of lattice spin systems
and can be considered a generalization of the MF and Bethe
approximations of the stationary density matrix.
It has proven to be extremely powerful in detecting the
existence of phase transitions, predicting their order and
their location in the parameter space, and computing correlation functions;
see, e.g.,
\cite{P05,CGP1996,CGTM1999}.

More specifically, the initial $M$-block probabilities are chosen, 
and then, at each step of the iteration, the $(m-1)$-block probabilities 
are computed via symmetrized marginalization (half of the sum of 
the left and right marginals). 
Next, the equations~\eqref{eq:it014}--\eqref{eq:it040} are 
employed to compute the new $M$-block probabilities. 
This algorithm is repeated until a fixed point 
is reached. In general, we fix an \emph{a priori} 
precision of $10^{-8}$ for the $L^2$-norm of the $M$-block probabilities.

For more explicit expressions of the iterative equation in the cases 
$M=3,5,7,9$ we refer the interested reader to \cite{CLS2024}. Here 
we note that, from
\eqref{eq:it030} and \eqref{eq:it040} we get
\begin{equation}
\label{appMF}
\varrho(x_1x_2x_3)
\approx
\frac{\varrho(x_1x_2)\varrho(x_2x_3)}{\varrho(x_2)}
\approx
\varrho(x_1)\varrho(x_2)\varrho(x_3)
,
\end{equation}
so that the $1$-BA is nothing but the MF approximation. 

Several papers (e.g., \cite{fuks2012,FF15,CLS2024}) have exploited the remarkable ability 
of the BA to reproduce accurately the 
structure of the stationary measures of \gls{pca}.
For instance, the 
paper \cite{CLS2024}  approached this 
problem not only in the diploid case, but also for a mixture 
of three local rules resulting in a \gls{pca} 
which exhibits a double transition 
as a function of the mixing parameter. 
We also refer to the papers discussed in the introductory part of this note.

\section{A highly symmetric diploid}
\label{s:diploid} 
In this paper we are interested in the diploid obtained as a 
mixture of the \gls{eca} 60 and 102, and show that, due to 
its highly symmetric character, the BA, at 
least at the order that we have considered, is not able to capture 
the existence of the zero-density stationary state.
The  behavior of the diploid is 
described in Table~\ref{tab:rules_60_102}. 
\begin{table}[H]
    \centering
    \begin{tabular}{|l|cccccccc|}
    \hline
         & 111 & 110 & 101 & 100 & 011 & 010 & 001 & 000\\
        \hline
        \textbf{60} & 0 & 0 & 1 & 1 & 1 & 1 & 0 & 0\\
        \textbf{102} & 0 & 1 & 1 & 0 & 0 & 1 & 1 & 0 \\
        \hline
    \end{tabular}
    \caption{Deterministic dynamics of ECA 60 and 102.
    The rule is reported in the first column. The following columns 
provide the update of the central site of the neighborhood reported 
in the first row.}
    \label{tab:rules_60_102}
\end{table}
\noindent The two rules are conjugate to each other  under left-right reflection.
In Figure~\ref{f:fig-evol} the evolution of the two deterministic rules is shown:  both rules produce the \emph{Sierpiński triangle} patterns from a single active cell. 
\begin{figure}

\includegraphics[width=0.15\textwidth, height=5cm]{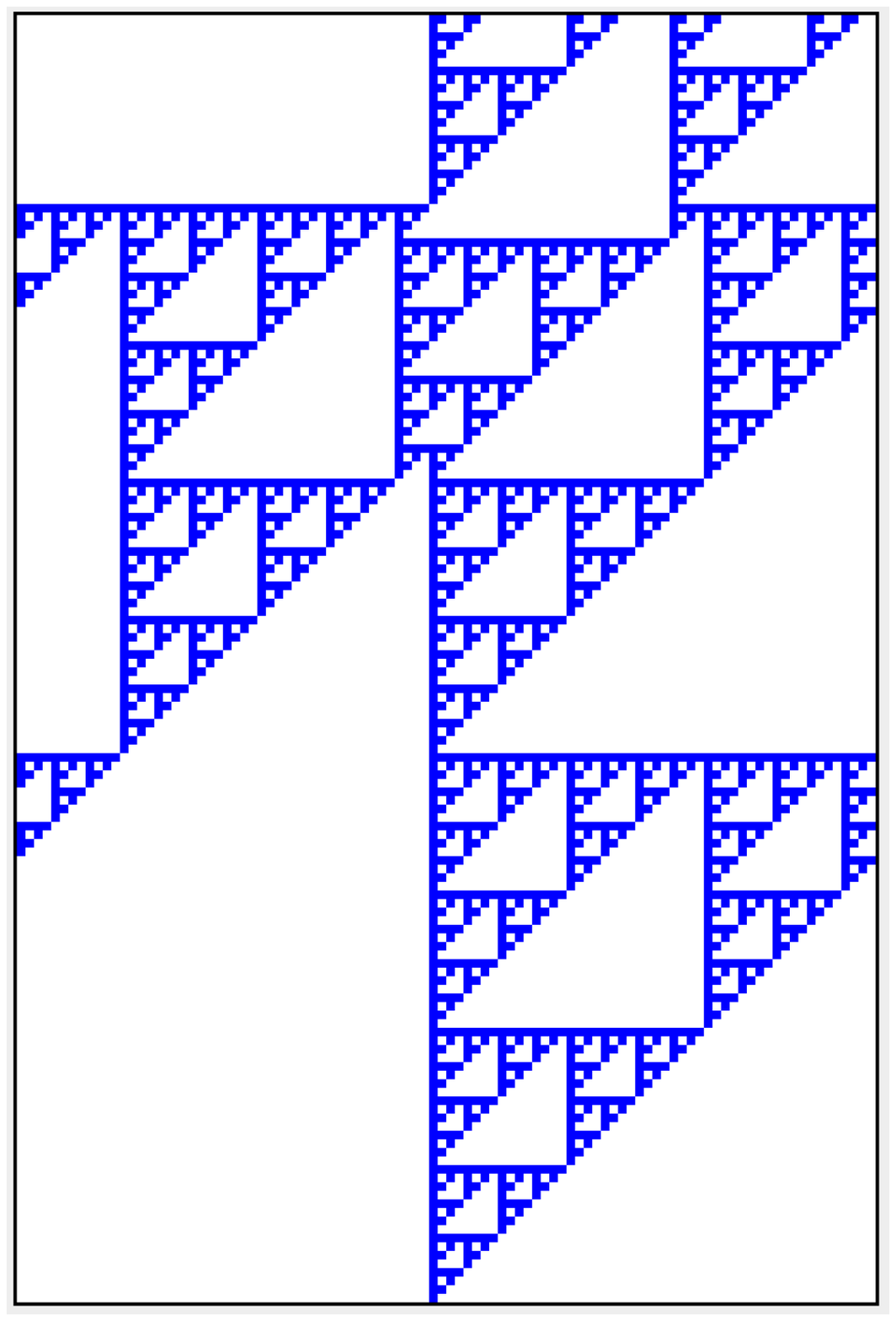}
\hskip 0.05 cm
\includegraphics[width=0.15\textwidth, height=5cm]{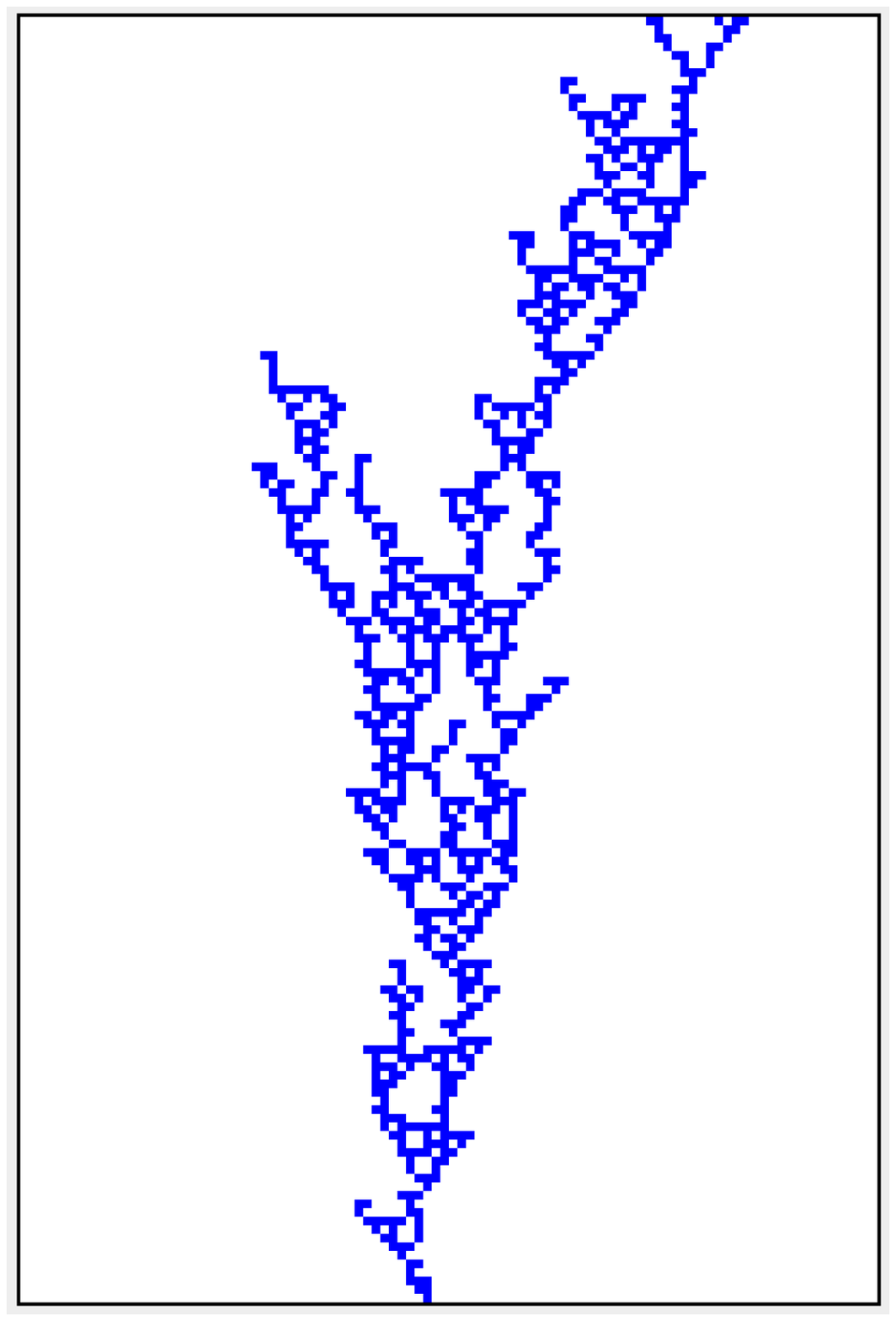}
\hskip 0.05 cm
\includegraphics[width=0.15\textwidth,height=5cm]{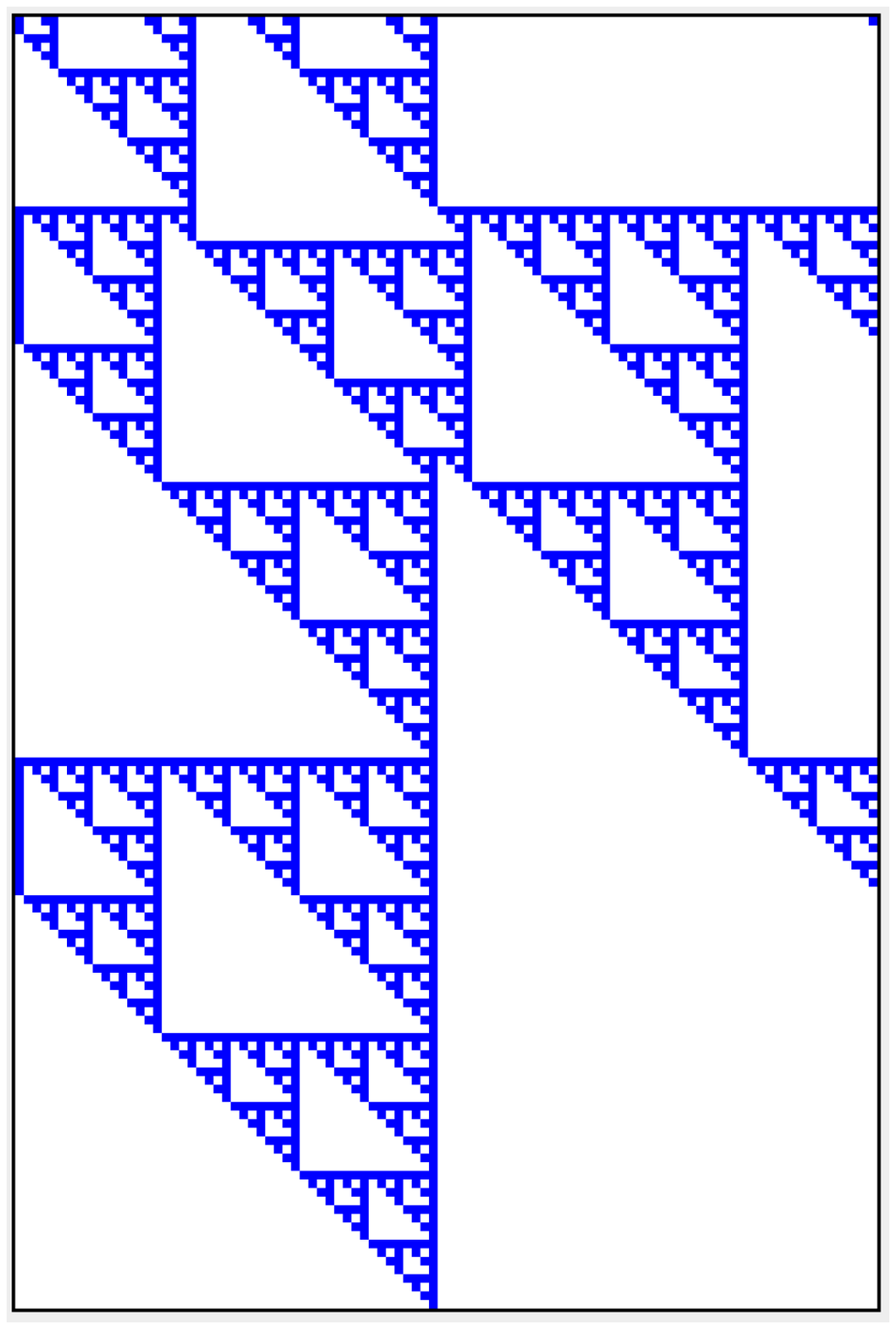}
\caption{The left and right panels report, respectively, the 
evolution up to time $150$ of the \gls{eca} 60 and 102 on a $100$-cell
periodic lattice; the central panel reports the diploid mixture with parameter $\lambda=1/2$. Unit squares represent cells in state one.
In the initial configuration the state of all the cells is zero, 
but for the one at the center of the lattice. In the plot the time runs upwards.
}
\label{f:fig-evol}
\end{figure}
\noindent Moreover, they correspond to well-known Boolean functions: ECA 60 implements the local rule $f_{1}(a,b,c)=a\oplus b$,
while ECA 102 implements $f_{2}(a,b,c)=b\oplus c$,
where $\xor$ denotes addition modulo 2 (XOR operator). This formal correspondence allows us to view these cellular automata as spatially distributed XOR functions. Furthermore, both rules are \emph{additive}: for any two configurations $x,y\in\{0,1\}^{\bbZ_n}$,
their sum (bit-wise XOR) evolves linearly:
$$
A(x\oplus y)=A x\oplus A y,
$$
where we denote by $A : X_n \to X_n$ the linear operator over $\mathbb{Z}_2$
associated for instance to the deterministic update of ECA~60, defined by:
\begin{equation}
\label{eq:def-A}
(Ax)_i := x_{i-1}\oplus x_i,
\qquad i\in\mathbb{Z}_n,
\end{equation}
where the indices are taken modulo $n$.
The linearity allows explicit computation of their space-time evolution and
leads to characteristic diagonal stripe patterns.
Moreover, on finite rings of odd length, both ECA $60$ and ECA ${102}$ are invertible, hence the corresponding dynamics are \emph{reversible}:
each configuration has a unique predecessor under the same rule.

In Figure~\ref{fig:dipsimo} we show the time evolution of the diploid with mixing parameter $\lambda=1/2$, and we observe that after $1080$ iterations the system reaches the absorbing state with zero density.
\begin{figure}
    \centering
    
    \subfigure[{$t\in[0,300]$}]{\includegraphics[width=0.22\textwidth]{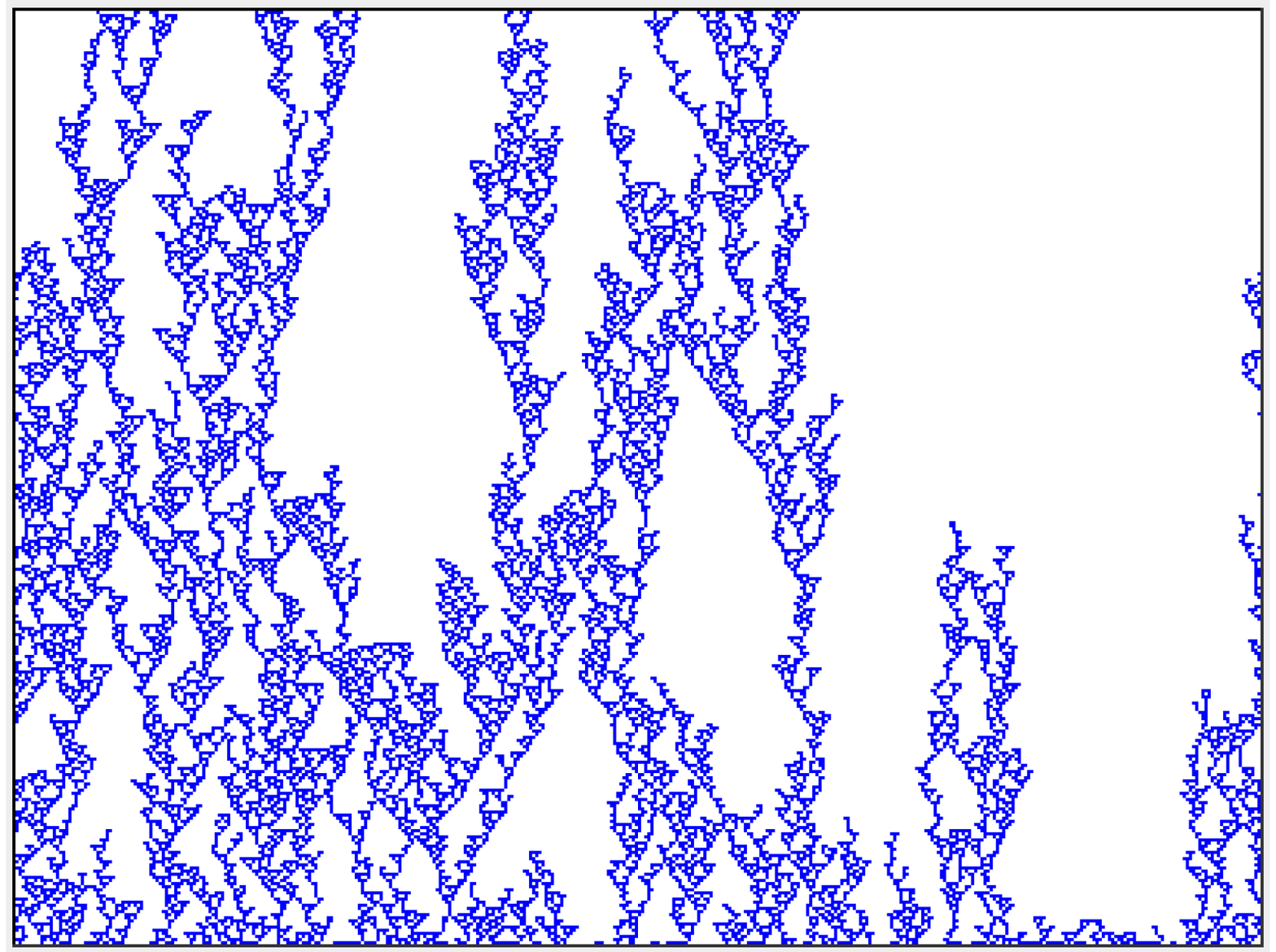}} 
    \subfigure[{$t\in[301,600]$}]{\includegraphics[width=0.22\textwidth]{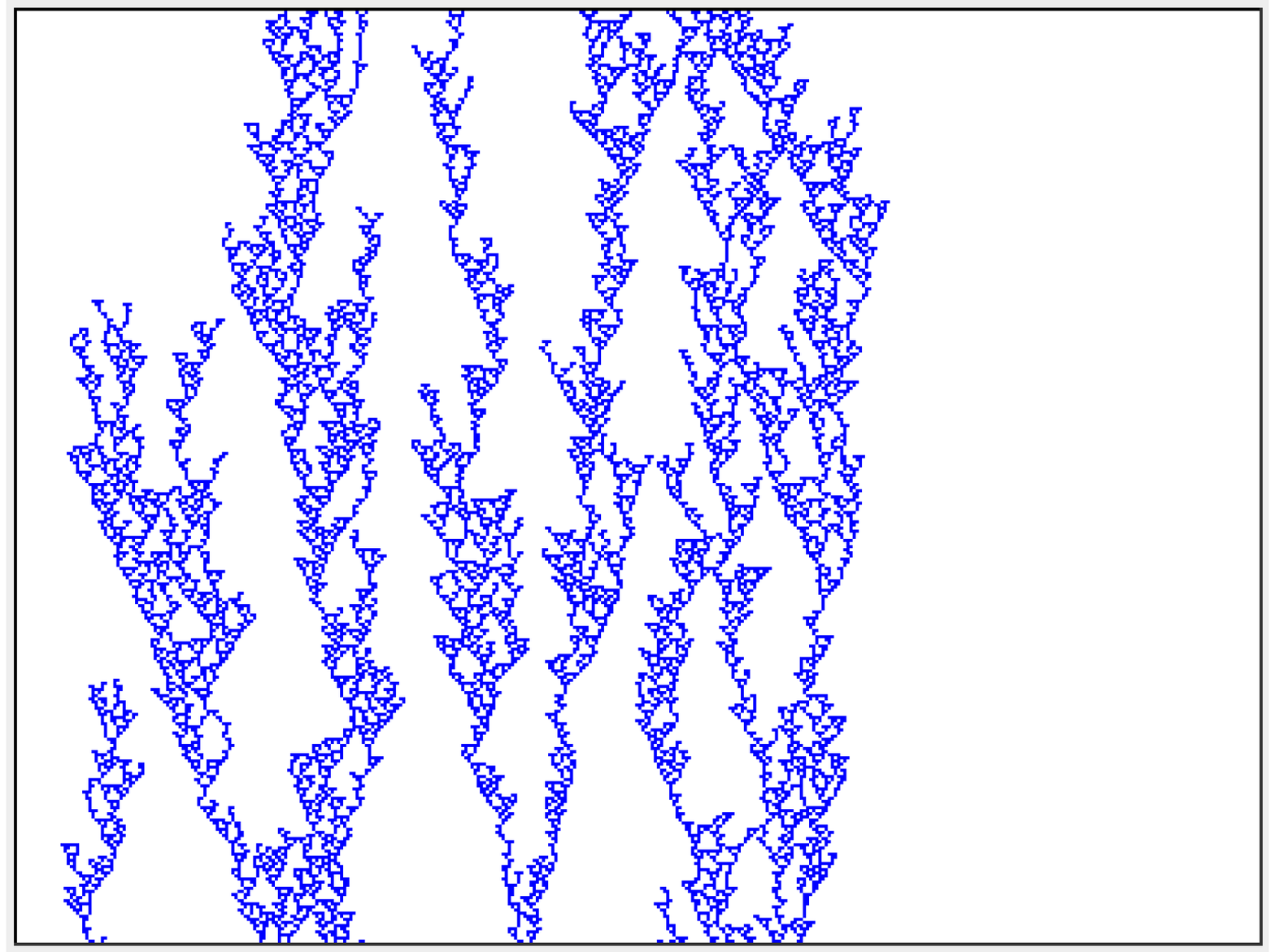}} 
    
    \subfigure[{$t\in[601,900]$}]{\includegraphics[width=0.22\textwidth]{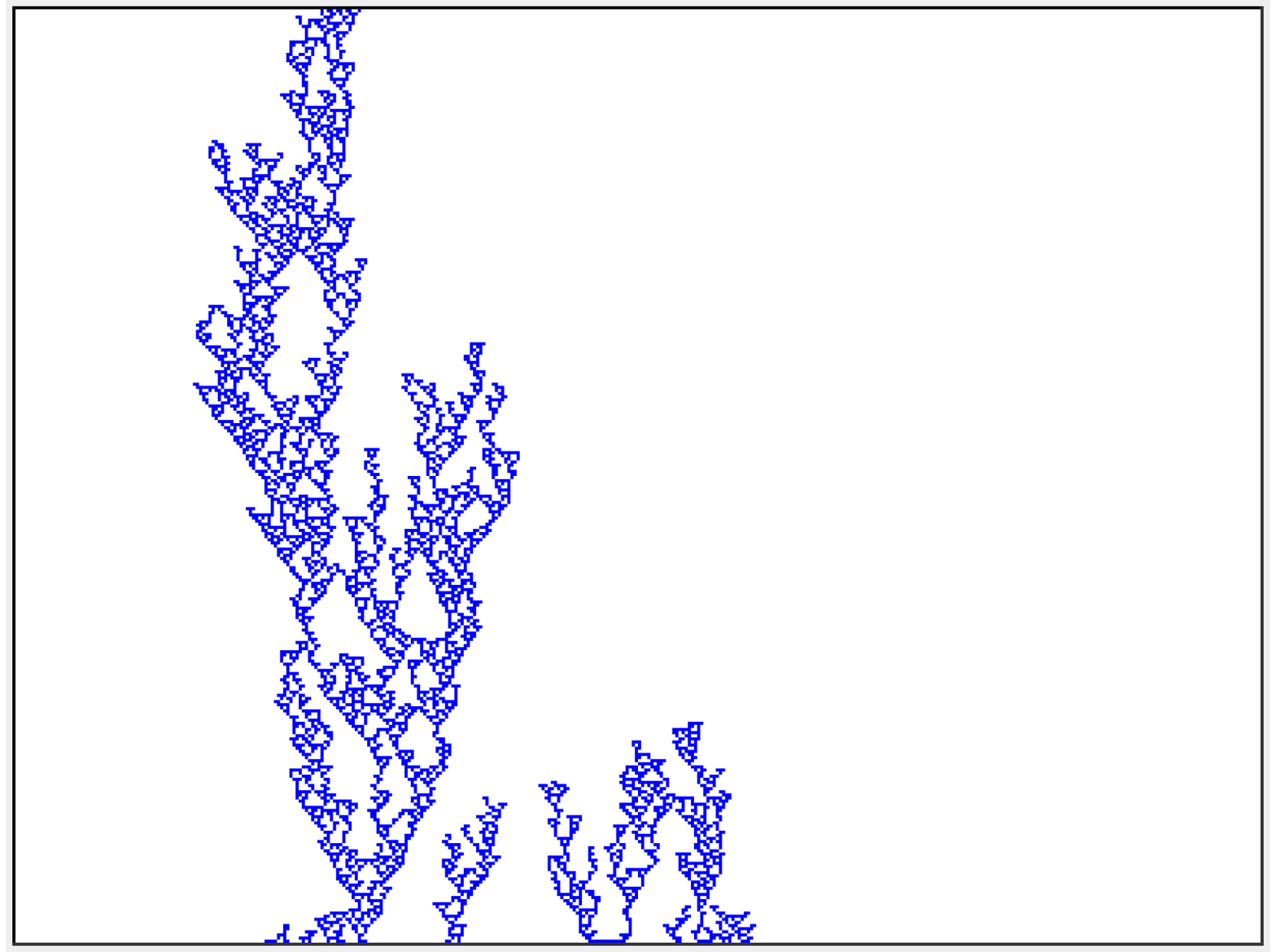}}
    \subfigure[{$t\in[901,1095]$}]{\includegraphics[width=0.22\textwidth]{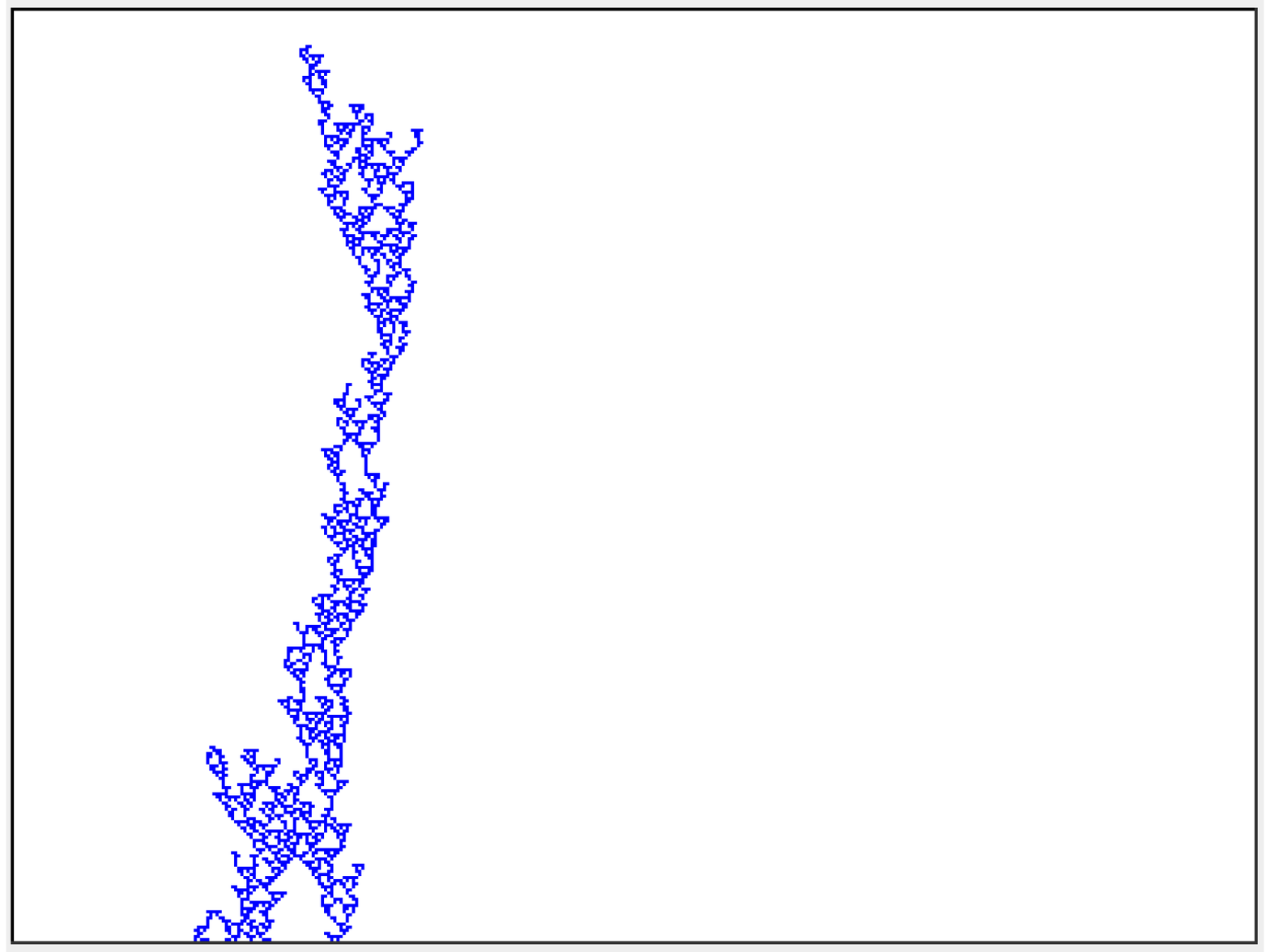}}
    \caption{Time evolution of the diploid with $400$ cells, mixture of rules 60 and 102, with parameter $\lambda=1/2$. The initial configuration is a realization of a Bernoulli distribution of cells with parameter $1/2$. The time evolves upwards (a) $t\in[0,300]$  (b) $t\in[301,600]$  (c) $t\in[601,900]$  (d) $t\in[901,1095]$. At time $t=1080$ the zero-density state is reached.}
    \label{fig:dipsimo}
\end{figure}
We consider now the late-time behavior of the diploid with a general $\lambda$, see \eqref{mod010bis}, 
with $f_1$ the rule \gls{eca} 60 and 
$f_2$ the rule \gls{eca} 102.
The Monte Carlo simulation for the density is
reported in Figure~\ref{f:fig-dens0}: the simulation was 
performed on the lattice with $n=10^5$ cells, for a total 
time of $10^5$. The density was estimated using \eqref{e:densityMC}
with $10^4$ as thermalization time and $10^2$ as decorrelation time.
According to the Monte Carlo results, the system undergoes two transitions 
between a zero and non-zero density stationary state, with a finite zero-density region in between. The presence of two
transition points ($\lambda_c^{l}$ and $\lambda_c^{r}$),
symmetric with respect to $1/2$, is 
an obvious consequence of the left-right symmetry of the model. In our simulation we have $\lambda_c^{l}\in(0.37,0.38)$ and $\lambda_c^{r}\in(0.62,0.63)$.
\begin{figure}
\centering
\includegraphics[width=0.47\textwidth,height=0.37\textwidth]{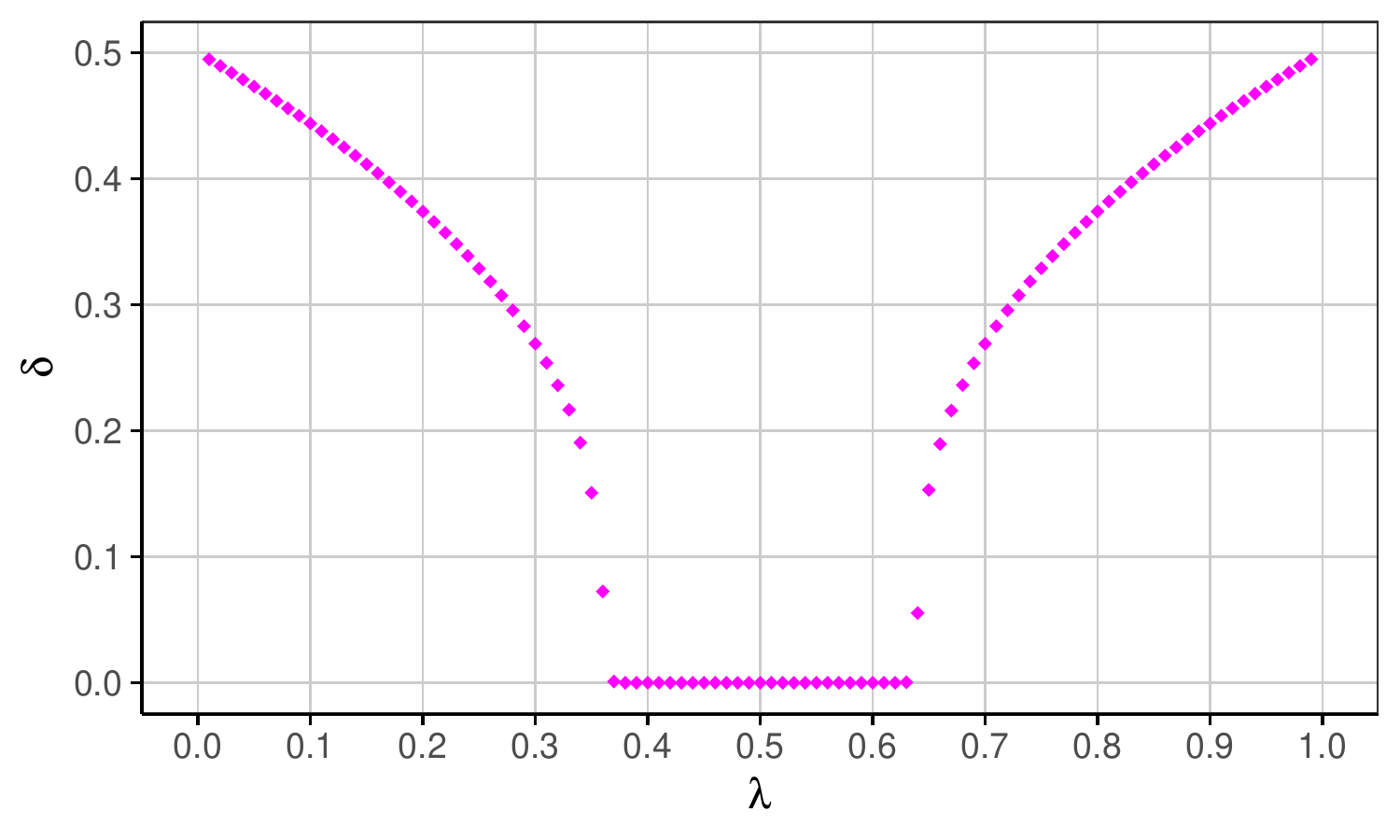}
\caption{Density as function of the mixing parameter $\lambda$ 
computed via Monte Carlo simulations for a lattice with $n=10^5$. The thermalization and decorrelation times were $10^4$ and $10^2$ respectively. 
}
\label{f:fig-dens0}
\end{figure}

\section{Rigorous results for the balanced case}
\label{s:rigorous} 
We consider here the diploid mixture of ECA 60 and 102 with $\lambda=1/2$ (see Figure~\ref{fig:dipsimo} for a realization of a path reaching the stationary absorbing state with zero density), whose relaxation towards the stationarity density can be studied rigorously. This diploid can indeed be  written as:
\begin{equation}
\label{e:cri}
\zeta_i^{t+1}= \begin{cases}\zeta_{i-1}^t \oplus \zeta_i^t & \text { if } \zeta_{i-1}^t=\zeta_{i+1}^t \\ \eta_{i, t+1} & \text { if } \zeta_{i-1}^t \neq \zeta_{i+1}^t\end{cases}
\end{equation}
where  $\eta_{i, t+1} \stackrel{i.i.d.} {\sim} \textnormal{Ber}(1/2)$. 
In other words, for $\lambda=1/2$ at each site and time step, 
 if the two outer neighbors are equal, the update acts deterministically,
      $\zeta_i^{t+1}=\zeta_{i-1}^t\oplus \zeta_i^t$;
if the two outer neighbors differ, the site updates randomly, $\zeta_i^{t+1}=\eta_{i,t+1}$ with $\eta_{i,t+1}\sim\mathrm{Ber}(1/2)$.

 We will prove that, on a finite periodic lattice, the all-zero absorbing 
configuration $\mathbf{0}$ is reached almost surely in finitely many steps.
Using the fact that the deterministic sites $D_t$ (see below) evolve 
according to the XOR rule and that to the complementary sites $F_t$ are assigned 
random bits, we will show that there exists a positive-probabiity 
sequence of realizations of the Bernoulli variables driving the system
to the absorbing state $\mathbf{0}$ in finitely many steps.    
Our goal is to show that there exists a positive-probability sequence of realizations of Bernoulli variables driving the system to the absorbing state $\mathbf{0}$.

For any fixed $t$  we define the sets:
\begin{equation}\label{eq:D-F}
D_t:=\{i\in\bbZ_n:\zeta^t_{i-1}=\zeta^t_{i+1}\},\qquad F_t:=\bbZ_n\setminus D_t.
\end{equation}
A configuration $y\in X_n$ is called \emph{forbidden} if there exists a site $i\in\bbZ_n$ such that $(y_{i-1},y_i,y_{i+1})\in \{(101),(010)\}$, otherwise it is called \emph{admissible}. With abuse of language, we also call a triple \emph{admissible} if it differs from $101$ and $010$.

%

We give now a sketch of the structure of the proof.
The first observation (Lemma~\ref{lem:det-no-forbidden}) is that a pair of deterministic sites (i.e., belonging to $D_t$) in a  triple can never produce
$010$ or $101$ under the PCA update.
The second ingredient is an admissible-triple construction on any finite
interval (Lemma~\ref{lem:interval}): we show that it is possible to choose the Bernoulli variables such that in one time step the evolved configuration is admissible. The third ingredient is a local cleaning principle (Lemma~\ref{lem:cleaning}): on any interval where all
triples are admissible, setting all Bernoulli random variables to zero cleans the entire
interval to zero in one step. Using these tools, we then perform a four-step construction.  In the first step,
we apply Lemma~\ref{lem:interval} at time $t$ to the large bulk interval $\{1,\dots,n-2\}$ to
ensure that all triples with centers $\{2,\dots,n-3\}$ are admissible at time
$t+1$.  In the second step, Lemma~\ref{lem:cleaning} is used to force all bulk
spins $\{2,\dots,n-3\}$ to zero at time $t+2$.  The only non-zero sites 
may lie in the four-site boundary block $\{n-2,n-1,0,1\}$.  At this point the boundary block has
the form $(0,a,b,c,d,0)$. This block is handled by a four-site boundary lemma (Lemma~\ref{lem:4site}), which observes that
the six-site boundary window $(0,a,b,c,d,0)$ is itself an interval with fixed
zero boundaries and can therefore be treated by another application of Lemma~\ref{lem:interval}.  This yields admissibility of all triples at
time $t+3$.  Finally, one last application of Lemma~\ref{lem:cleaning} forces the
entire configuration to zero at time $t+4$. Altogether, this shows that $\mathbf{0}$ is reachable in at most four steps
with strictly positive probability, establishing Theorem~\ref{thm:4step}.

We now prove the lemmas needed for the four–step absorption scheme.


\begin{lemma}
\label{lem:det-no-forbidden}

Consider the evolved configuration $\zeta^{t}$ at time $t$ and the site 
$i\in\bbZ_n$. If at least two of the sites $i-1,i,i+1$ belong to $D_t$, then, for every
choice of the Bernoulli variables at time $t+1$, the triple
$(\zeta^{t+1}_{i-1},\zeta^{t+1}_i,\zeta^{t+1}_{i+1})$ is admissible.

\end{lemma}

\begin{proof}

Assume that at least two of the sites $(i-1,i,i+1)$ are deterministic at time $t$. 
Due to the left-right symmetry, there are only three possible configurations
of deterministic sites.
For each of them, the deterministic updates impose a
forced algebraic relation among the components of the triple
$(\zeta^{t+1}_{i-1},\zeta^{t+1}_i,\zeta^{t+1}_{i+1})$, independently of the Bernoulli variables. These relations are summarized in
Table~\ref{tab:deterministic_triples}.
In all three cases, the forced relation is incompatible with the patterns
$010$ and $101$:
the first two cases force equality of two adjacent entries, while the third
forces inequality of the two outer entries, whereas both forbidden patterns
require $\zeta^{t+1}_{i-1}\neq \zeta^{t+1}_i=\zeta^{t+1}_{i+1}$ or $\zeta^{t+1}_{i-1}=\zeta^{t+1}_{i+1}\neq \zeta^{t+1}_i$.

If all three sites $i-1,i,i+1$ belong to $D_t$, then one of the first two cases
of Table~\ref{tab:deterministic_triples} applies. We briefly justify the forced relations reported in
Table~\ref{tab:deterministic_triples}.
If $\{i-1,i\}\subseteq D_t$, then $\zeta^t_{i-2}=\zeta^t_i$ and the deterministic
updates give
\[
\zeta^{t+1}_{i-1}=\zeta^t_{i-2}\xor\zeta^t_{i-1}
=\zeta^t_i\xor\zeta^t_{i-1}
=\zeta^{t+1}_i,
\]
so the first two entries of the triple coincide.
The case $\{i,i+1\}\subset D_t$ is completely analogous by left-right symmetry.
Finally, if $\{i-1,i+1\}\subseteq D_t$ and $i\in F_t$, then
\[
\zeta^{t+1}_{i-1}=\zeta^t_i\xor\zeta^t_{i-1},\qquad
\zeta^{t+1}_{i+1}=\zeta^t_i\xor\zeta^t_{i+1},
\]
and since $i\in F_t$ implies $\zeta^t_{i-1}\neq\zeta^t_{i+1}$, XOR with the same
bit $\zeta^t_i$ preserves inequality, yielding
$\zeta^{t+1}_{i-1}\neq\zeta^{t+1}_{i+1}$.
In all cases, the resulting relations are incompatible with the forbidden
triples $010$ and $101$.
\end{proof}

\begin{table}[t]
\centering
\setlength{\tabcolsep}{8pt}
\renewcommand{\arraystretch}{1.2}
\begin{tabular}{ll}
\toprule
Deterministic sites at time $t$
&
Forced relation 
\\
\midrule
$\{i-1,i\}\subseteq D_t$
&
$\zeta^{t+1}_{i-1}=\zeta^{t+1}_i$
\\

$\{i,i+1\}\subseteq D_t$
&
$\zeta^{t+1}_i=\zeta^{t+1}_{i+1}$
\\

$\{i-1,i+1\}\subseteq D_t$, $i\in F_t$
&
$\zeta^{t+1}_{i-1}\neq\zeta^{t+1}_{i+1}$
\\
\bottomrule
\end{tabular}
\caption{Forced algebraic relations on the updated triple
$(\zeta^{t+1}_{i-1},\zeta^{t+1}_i,\zeta^{t+1}_{i+1})$ when at least two sites
among $\{i-1,i,i+1\}$ are deterministic at time $t$.
Each relation excludes the triples $010$ and $101$,
independently of the Bernoulli variables.}
\label{tab:deterministic_triples}
\end{table}

\begin{lemma}
\label{lem:interval}
Fix $n\ge 5$ and consider the evolved configuration  $\zeta^t$.  Consider the interval
\[
I:=\{1,2,\dots,n-2\}\subset \mathbb{Z}_n.
\]
Then there exists a choice of the Bernoulli variables at time $t+1$ such that all triples
$(\zeta^{t+1}_{i-1},\zeta^{t+1}_i,\zeta^{t+1}_{i+1})$ are admissible for every $i\in\{2,\dots,n-3\}$.
\end{lemma}

\begin{proof}
 
Let $I:=\{1,2,\dots,n-2\}$.  
We construct the variables $y_1,\dots,y_{n-2}$ sequentially from the left to the right.

Starting from $i=1$, we distinguish two cases. 
If $1\in D_t$, set $y_1:=\zeta^t_{0}\oplus\zeta^t_{1}$.
For $1\in F_t$, choose $y_1$ arbitrarily (e.g.,\ $y_1:=0$).

If  $i\in\{2,3,\dots,n-2\}$, we define $y_i$ by:
\[
y_i :=
\begin{cases}
\zeta^t_{i-1}\oplus\zeta^t_{i}, & \text{if } i\in D_t,\\[1mm]
y_{i-1}, & \text{if } i\in F_t.
\end{cases}
\]
For each $i\in\{3,4,\dots,n-2\}$, we show that the triple   $(y_{i-2},y_{i-1},y_i)$ centered at
$i-1$  is admissible; this will imply the claim.

Fix $i\in\{3,\dots,n-2\}$ and consider the triple $(y_{i-2},y_{i-1},y_i)$.
We distinguish three cases.

\smallskip
\noindent\emph{Case 1: $i\in F_t$.}
Then $y_i=y_{i-1}$ by the algorithm, hence the last two entries of the triple coincide. Therefore, the triple is admissible.

\smallskip
\noindent\emph{Case 2: $i\in D_t, i-1\in F_t$}. 
Therefore, the algorithm set $y_{i-1}=y_{i-2}$,
so that the first two entries of the triple centered at $i-1$ coincide. Hence, it is admissible.\\
\noindent\emph{Case 3: $i\in D_t, i-1\in D_t$}.
Thus, by
Lemma~\ref{lem:det-no-forbidden} the triple centered at $i-1$ is admissible.
\end{proof}
\begin{remark}\label{rem:forced-forbidden}
We would like to stress that  \emph{locally} we can easily prove that at every site
$i$ there exists at least one choice of the local Bernoulli variable that makes the triple centered at $i$ admissible at time $t\!+\!1$.

Fix, indeed, a site $i\in\bbZ_n$.  
Suppose that, at time $t+1$, the triple
$(\zeta^{t+1}_{i-1},\,\zeta^{t+1}_i,\,\zeta^{t+1}_{i+1})
$
were equal to $010$ or  $101$ with probability~$1$ over the randomness at
time~$t$. Then each of its three coordinates must be \emph{deterministically} fixed by
the update rule~\eqref{e:cri}.
For instance, if the triple is almost surely $010$, then necessarily
\begin{align*}
\zeta^{t+1}_{i-1}=0, \qquad 
\zeta^{t+1}_i=1, \qquad
\zeta^{t+1}_{i+1}=0.
\end{align*}
Examining the update rule, this implies:
\begin{enumerate}
\item the site $i-1$ must be deterministic at time $t$, and its neighbours satisfy
$(\zeta^t_{i-2},\zeta^t_{i})\in\{(0,0),(1,1)\}$;
\item the site $i$ must also be deterministic at time $t$, but its neighbours must satisfy
$(\zeta^t_{i-1},\zeta^t_{i+1})\in\{(0,1),(1,0)\}$;
\item similarly $i+1$ must satisfy $(\zeta^t_{i},\zeta^t_{i+2})\in\{(0,0),(1,1)\}$.
\end{enumerate}
These conditions are incompatible:  
from the first item we obtain $\zeta^t_{i-1}=\zeta^t_i$, while from the second item we require $\zeta^t_{i-1}\neq\zeta^t_i$.  
Thus the preimage of a triple which is almost surely $010$ (or $101$) is empty.

 However, this local argument does \emph{not} imply Lemma~\ref{lem:interval}, since it does not guarantee that all triples can be made admissible \emph{simultaneously}.
Indeed, triples overlap, so the admissible choices at different sites may interact.  Lemma~\ref{lem:interval} shows indeed that, despite these global compatibility constraints, a full assignment of the Bernoulli variables
making every triple admissible always exists.
\end{remark}

\begin{lemma}
\label{lem:cleaning}
Assume $\zeta^t=x$, with $x$ a configuration with all admissible triples. If we choose $\eta_{j,t+1}=0$ whenever $j\in F_t$, then
\[
\zeta^{t+1} = \mathbf{0}.
\]
\end{lemma}

\begin{proof}
If $j\in D_t$, admissibility forces $x_{j-1}=x_j=x_{j+1}$, and hence
$\zeta^{t+1}_j=x_{j-1}\oplus x_j=0$.
If $j\in F_t$, setting $\eta_{j,t+1}=0$ gives $\zeta^{t+1}_j=0$ by definition.
\end{proof}

\begin{lemma}
\label{lem:4site}
Let $z\in X_n$ be a configuration such that $z_i=0$ for all $i \in \{2,3,\dots,n-3\}$.
If $\zeta^t=z$, then there exists a choice of Bernoulli variables at time $t+1$ such that the evolved 
configuration $\zeta^{t+1}$ is such that all the triples centered at $n-2,n-1,0,1$ are admissible.

\end{lemma}

\begin{proof}
Consider the six-site block
\[
(0,a,b,c,d,0)=:(z_{n-3},z_{n-2},z_{n-1},z_0,z_1,z_2).
\]
Regard it as an interval of size $6$ and apply Lemma~\ref{lem:interval} to $I=\{n-3,n-2,n-1,0,1,2\}$ with fixed boundaries $z_{n-3}=z_2=0$.
Lemma~\ref{lem:interval} guarantees that all triples centered not in the boundary are admissible.
\end{proof}


\begin{theorem}
\label{thm:4step}
For any $n\geq5$, consider the configuration $\zeta^t$ at time $t$, then there exists a choice of the Bernoulli
variables over the next four times such that $\zeta^{t+4}\equiv\mathbf{0}$.
In other words, the all-zero configuration is reachable with strictly
positive probability in four steps.
\end{theorem}

\begin{proof}
We provide a proof organized in four steps.\\
\emph{Step 1 (interval admissibility on the bulk).}
Apply Lemma~\ref{lem:interval} to the interval $\{1,2,\dots,n-2\}$, and choose the Bernoulli variables so that
$\zeta^{t+1}$ is admissible on all triples centered at
$\{2,\dots,n-3\}$.

\emph{Step 2 (local cleaning of the bulk).}
Choose $\eta_{i,t+2}=0$ for all $i\in\{2,\dots,n-3\}\cap F_{t+1}$.
By Lemma~\ref{lem:cleaning}, this yields
\[
\zeta^{t+2}_i = 0 \quad\text{for all } i\in\{2,\dots,n-3\}.
\]

\emph{Step 3 (boundary cleaning).}
The configuration $\zeta^{t+2}$ has the form
\[
(0,a,b,c,d,0)
\quad\text{on }\{n-3,n-2,n-1,0,1,2\},
\]
and zero otherwise.
Apply Lemma~\ref{lem:4site} at time $t+2$.
This yields $\zeta^{t+3}$ has only admissible triples.

\emph{Step 4 (global cleaning).}
Since $\zeta^{t+3}$ has only admissible triples, setting all
$\eta_{i,t+4}=0$, for all $i\in F_{t+4}$, yields $\zeta^{t+4}_i=0$ for all $i$ by
Lemma~\ref{lem:cleaning}.
Hence, $\zeta^{t+4}\equiv\mathbf{0}$.

All four random variable assignments have positive probability, therefore the
probability of this four-step path to $\mathbf{0}$ is strictly positive.
\end{proof}

Before stating Theorem~\ref{cor:aas}, we stress that its proof is standard in the theory of
finite-state Markov chains. For a Markov chain with a unique absorbing state, if this state
is reachable from every other configuration with strictly positive probability, then absorption
occurs almost surely in finite time. We nevertheless provide a proof for the sake of completeness and to make explicit the role of
Theorem~\ref{thm:4step} in the almost sure absorption.

\begin{theorem}\label{cor:aas}
For $\lambda=1/2$, and $n\geq 5$, the configuration $\mathbf{0}$ is the unique absorbing
state of the PCA \eqref{e:cri}, and it is reached almost surely in finite time
from any initial configuration.
\end{theorem}

\begin{proof}
By Theorem~\ref{thm:4step}, for every initial configuration $x\in X_n$ we have
\[
\mathbb P_x(\zeta^{4}=\mathbf 0)>0.
\]
Since the state space $X_n$ is finite, the minimum of these strictly positive numbers
is strictly positive; define
\[
\varepsilon \;:=\; \min_{x\in X_n}\mathbb P_x(\zeta^{4}=\mathbf 0)\;>\;0.
\]
For $k\ge 0$ set $A_k:=\{\zeta^{4k}\neq \mathbf 0\}$.
We claim that for every $x\in X_n$ and every $k\ge 0$:
\begin{equation}\label{eq:geom-rec}
\mathbb P_x(A_{k+1})\le (1-\varepsilon)\,\mathbb P_x(A_k).
\end{equation}
To prove \eqref{eq:geom-rec}, we use the law of total probability with respect to the
random variable $\zeta^{4k}$:
\begin{equation}\label{eq:total-prob}
\mathbb P_x(A_{k+1})
=
\sum_{y\in X_n}\mathbb P_x(A_{k+1}\mid \zeta^{4k}=y)\,\mathbb P_x(\zeta^{4k}=y).
\end{equation}
If $y=\mathbf 0$, then we have
$\mathbb P_x(A_{k+1}\mid \zeta^{4k}=\mathbf 0)=0$.
If $y\neq \mathbf 0$, then by the Markov property and time-homogeneity we have:
\[
\mathbb P_x(A_{k+1}\mid \zeta^{4k}=y)
=
\mathbb P_y(\zeta^{4}\neq \mathbf 0)
\le 1-\varepsilon.
\]
Thus
\begin{eqnarray}
\mathbb P_x(A_{k+1})
&=&
\sum_{y\neq \mathbf 0}\mathbb P_x(A_{k+1}\mid \zeta^{4k}=y)\,\mathbb P_x(\zeta^{4k}=y)\nonumber \\
&\le&
(1-\varepsilon)\sum_{y\neq \mathbf 0}\mathbb P_x(\zeta^{4k}=y). \nonumber
\end{eqnarray}
But $\sum_{y\neq \mathbf 0}\mathbb P_x(\zeta^{4k}=y)=\mathbb P_x(\zeta^{4k}\neq \mathbf 0)=\mathbb P_x(A_k)$,
so \eqref{eq:geom-rec} follows. Iterating \eqref{eq:geom-rec} yields, for all $k\ge 0$,
\begin{equation}\label{eq:geom-bound}
\mathbb P_x(A_k)\le (1-\varepsilon)^k.
\end{equation}
Consider now the event that the chain never hits $\mathbf 0$:
\[
\{\tau_{\mathbf 0}=\infty\}=\{\zeta^t\neq \mathbf 0\ \text{for all }t\ge 0\},
\]
where $\tau_{\mathbf 0}$ is the first hitting time of the configuration $\mathbf{0}$.
Since $\mathbf 0$ is absorbing, if the chain ever hits $\mathbf 0$ at some time $t$
then it stays there forever, in particular $\zeta^{4k}=\mathbf 0$ for all $k$ large enough.
Equivalently, the chain never hits $\mathbf 0$ if and only if
$\zeta^{4k}\neq \mathbf 0$ for infinitely many $k$, that is,
\begin{equation}\label{eq:tau-limsup}
\{\tau_{\mathbf 0}=\infty\}
=
\limsup_{k\to\infty}A_k
=
\bigcap_{m\ge 0}\ \bigcup_{k\ge m}A_k.
\end{equation}
Fix $m\ge 0$. From the inclusion
\[
\bigcap_{r\ge 0}\ \bigcup_{k\ge r}A_k \subseteq\bigcup_{k\ge m}A_k,
\]
monotonicity of probability gives
\[
\mathbb P_x(\tau_{\mathbf 0}=\infty)
=
\mathbb P_x\!\Big(\bigcap_{r\ge 0}\ \bigcup_{k\ge r}A_k\Big)
\le
\mathbb P_x\!\Big(\bigcup_{k\ge m}A_k\Big).
\]
Applying the union bound and \eqref{eq:geom-bound},
\[
\mathbb P_x\!\Big(\bigcup_{k\ge m}A_k\Big)
\le \sum_{k\ge m}\mathbb P_x(A_k)
\le \sum_{k\ge m}(1-\varepsilon)^k
= \frac{(1-\varepsilon)^m}{\varepsilon}.
\]
Since the right-hand side tends to $0$ as $m\to\infty$, we conclude that
\[
\mathbb P_x(\tau_{\mathbf 0}=\infty)=0,
\qquad\text{i.e.}\qquad
\mathbb P_x(\tau_{\mathbf 0}<\infty)=1.
\]
Thus the hitting time of $\mathbf 0$ is finite almost surely.

Assume by contradiction that there exists $y\in X_n$, $y\neq \mathbf 0$,
which is also absorbing. Then the chain started from $y$ remains at $y$ forever,
so $\mathbb P_y(\zeta^4=\mathbf 0)=0$. This contradicts Theorem~\ref{thm:4step}, which
implies $\mathbb P_y(\zeta^4=\mathbf 0)\ge \varepsilon>0$. Hence $\mathbf 0$
is the unique absorbing state.
\end{proof}

\section{Mean Field failure and block approximation inefficiency}
\label{s:failure}
Theorem~\ref{cor:aas} shows that for the diploid with ECAs 60, 102 with $\lambda=1/2$,  the all-zero configuration is the unique absorbing state of the PCA and is reached almost surely in finite time from any initial condition. Next, we want to move away from $\lambda=1/2$ and study the extended region with absorbing state. However, we will show that both the MF and the BA have difficulties detecting this absorbing stationary state and the transitions
to the finite-density state observed at the left and at the right of $\lambda=1/2$.

More precisely, if we consider now 
the MF approximation, it indeed fails completely in detecting the presence of the zero stationary state. In fact, by
recalling \eqref{mod010bis} and Table~\ref{tab:rules_60_102},
by writing \eqref{eq:it014} for the $1$-cell block $B_{1,1}$ and 
exploiting the approximation \eqref{appMF}, one gets 
\begin{align*}
\delta
=
\varrho(1)
=
&
\lambda(1-\delta)\delta^2
+
(1-\delta)\delta^2
+
(1-\lambda)(1-\delta)^2\delta
\\
&
+
(1-\lambda)(1-\delta)\delta^2
+
(1-\delta)^2\delta
+
\lambda(1-\delta)^2\delta
,
\end{align*}
which, besides the unstable solution $\delta=0$ (see Figure~\ref{f:fig-dens2}) admits the non-zero 
stable solution $\delta=1/2$ for every value of $\lambda\in(0,1)$. 
In other words, the MF approximation predicts the existence 
of both the zero density and the non-zero density solution 
for every $\lambda$, in contrast with the numerical results of Figure~\ref{f:fig-dens0}.

Furthermore, we also approach the problem with the BA and 
consider the cases $M=3,5,7,9,11,13$. 
As shown in 
Figure~\ref{f:fig-dens2}
the BA greatly improves the MF result, 
but, nevertheless, it is not able to capture the existence 
of the transition.

To be more precise, the measure concentrated on the zero block is
trivially a solution of the block approximation iterative equations, 
nevertheless, solving the iterative equations starting from the $1/2$ 
Bernoulli measure the fixed point found in the measure space 
has positive density for every $\lambda$, as reported in 
Figure~\ref{f:fig-dens2}. 
\begin{figure}
\centering
\includegraphics[width=0.47\textwidth,height=0.37\textwidth]{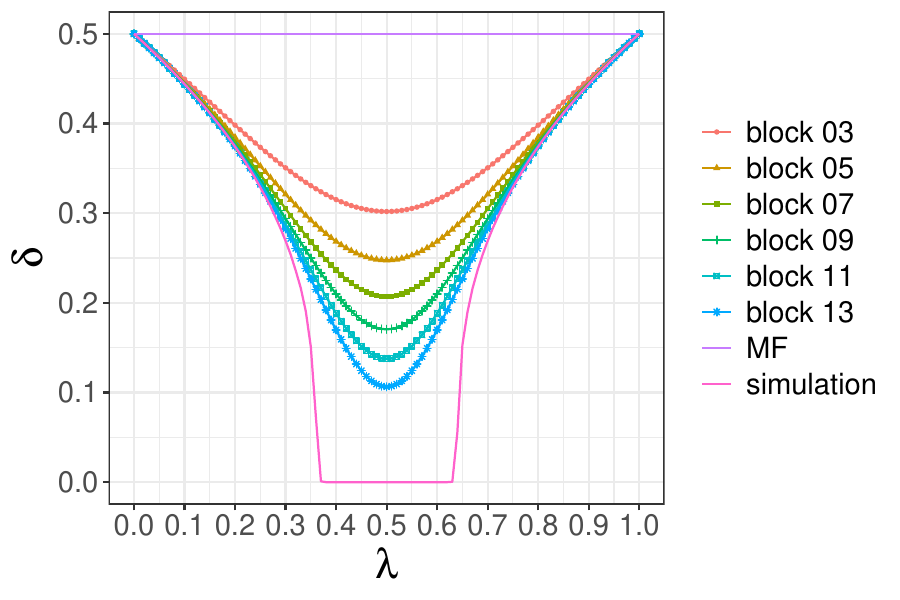}
\caption{Density as function of the mixing parameter $\lambda$ 
computed via MF, BA of different sizes, 
and 
Monte Carlo simulations.
For the BA we report the solution found 
initializing the iterative equations with the $1/2$ Bernoulli measure.
}
\label{f:fig-dens2}
\end{figure}
Thus, the BA method is 
not able to predict a transition and provide an estimate of the 
transition point. 
We have also explored, without success, 
different starting points for the iterative equations, looking for 
smart initial conditions which could, at least for specific values 
of $\lambda$, end up in the 
measure concentrated on the zero block so that 
the transition point could be estimated. 
Summarizing, the measure concentrated on the zero block 
is a stationary point of the iterative equations, but it 
appears as a highly unstable fixed point. As soon as the 
iterations are started from a slightly different point in the 
block measure space, the algorithm converges to the non-zero 
density solutions. 
This behavior is much different from what we observed in 
several cases, for instance for the diploid considered in 
\cite{CLS2024}, for which there exist values of $\lambda$ 
such that 
the measure concentrated on zero becomes a stable 
fixed point of the iterative equations 
providing a clear strategy to estimate the transition point. 
\begin{figure}
\centering
\includegraphics[width=0.4799\textwidth]{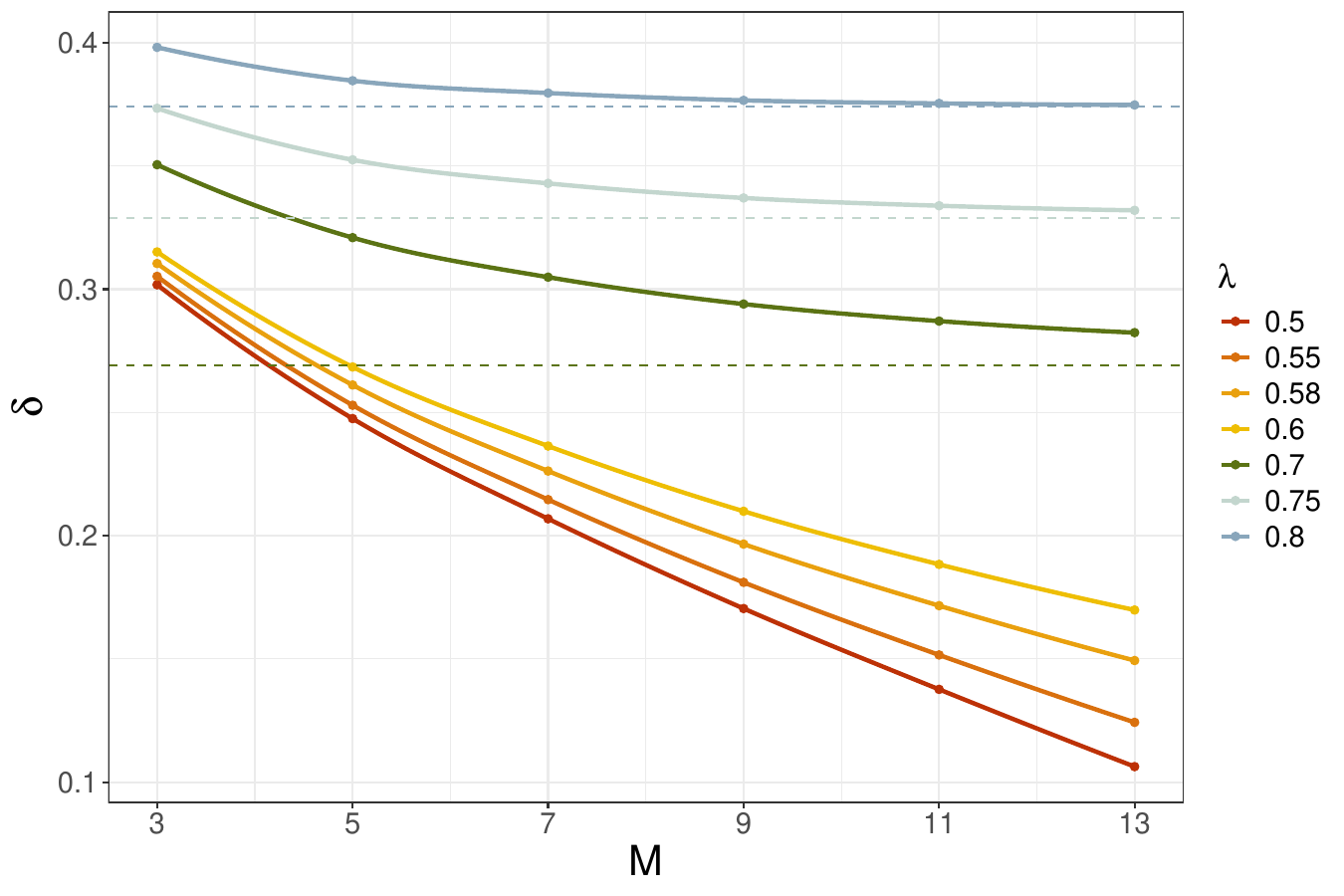}
\caption{Decay of the density as function of the size $M$ of the BA for different $\lambda$. The dashed lines denote the corresponding asymptotic densities of the Monte Carlo simulations. For large $\lambda$ the BA agrees with the simulations already at $M=13$, while it is off in the region of the zero-density absorbing state.}
\label{f:fig-blocksize}
\end{figure}

To further analyze this behavior, in  Figure~\ref{f:fig-blocksize} we show how the block
approximation $\rho_{M}(\lambda)$ depends on the block size
$M\in\{3,5,7,9,11,13\}$ for each fixed value of $\lambda$.  The sequence is
monotone decreasing and appears to converge as $M$ increases, suggesting that
the finite-block corrections follow an exponential relaxation.  Motivated by
this observation, we fit the six values $\rho_M(\lambda)$ with the
three–parameter form
\begin{equation}
   \rho^{(1)}_{M}(\lambda)
   = a(\lambda)\, e^{-\,b(\lambda)\, M} + c(\lambda),
   \label{eq:expfit_block}
\end{equation}
where $c(\lambda)$ represents the asymptotic BA that would be
obtained in the formal limit $M\to\infty$.  In this fit, $a(\lambda)$ measures
the amplitude of the finite–block corrections and $b(\lambda)$ the rate at
which these corrections decay.
For each $\lambda$, the parameters $(a(\lambda),b(\lambda),c(\lambda))$ were
determined by nonlinear least squares.
To ensure physical consistency of
the density, we imposed the natural positivity constraints
$a(\lambda)\ge 0$, $b(\lambda)\ge 0$ and $c(\lambda)\ge 0$, preventing the fit
from producing unphysical negative densities.  
The goodness of the fit was quantified using the coefficient of determination
\[
 R^2(\lambda)
 = 
 1 - 
 \frac{\sum_{M} 
      \left(\rho^{(1)}_M(\lambda)
      - a(\lambda)e^{-b(\lambda)M}-c(\lambda)\right)^2}
        {\sum_{M}
        \left(\rho^{(1)}_M(\lambda)-\overline{\rho^{(1)}(\lambda)}\right)^2},
\]
where $\overline{\rho^{(1)}(\lambda)}$ is the mean of the six values
$\rho^{(1)}_{M}(\lambda)$.  The constrained exponential fit achieves an excellent
description of the data: for all $\lambda\ge 0.01$ we find
\[
   0.93 \;\le\; R^2(\lambda) \;\le\; 0.9995,
\]
with $R^2(\lambda)>0.995$ throughout the region of interest around
$\lambda=1/2$.  The exponential law \eqref{eq:expfit_block} therefore provides
a quantitatively accurate interpolation of the finite-block densities and a
reliable extrapolation to $M\to\infty$.

The resulting asymptotic densities $c(\lambda)$ are shown in
Figure~\ref{fig:power_exp}, together with the Monte Carlo stationary densities
(black curve).   The agreement is excellent for
$\lambda$ also away from the  value $1/2$, confirming that the BA converges smoothly with $M$.  More importantly, in a narrow
region around $\lambda = 1/2$ the fitted asymptotic density becomes extremely
small and forms a plateau that is statistically compatible with
$c(\lambda)=0$.  This is entirely consistent with the exact absorbing behavior
of the PCA at $\lambda=1/2$: although the finite-block values
$\rho_M(1/2)$ remain positive for $M\le 13$, their exponential extrapolation
reveals that the infinite-block prediction is a vanishing density.
\begin{figure}
\centering
\includegraphics[width=0.45\textwidth,height=0.35\textwidth]{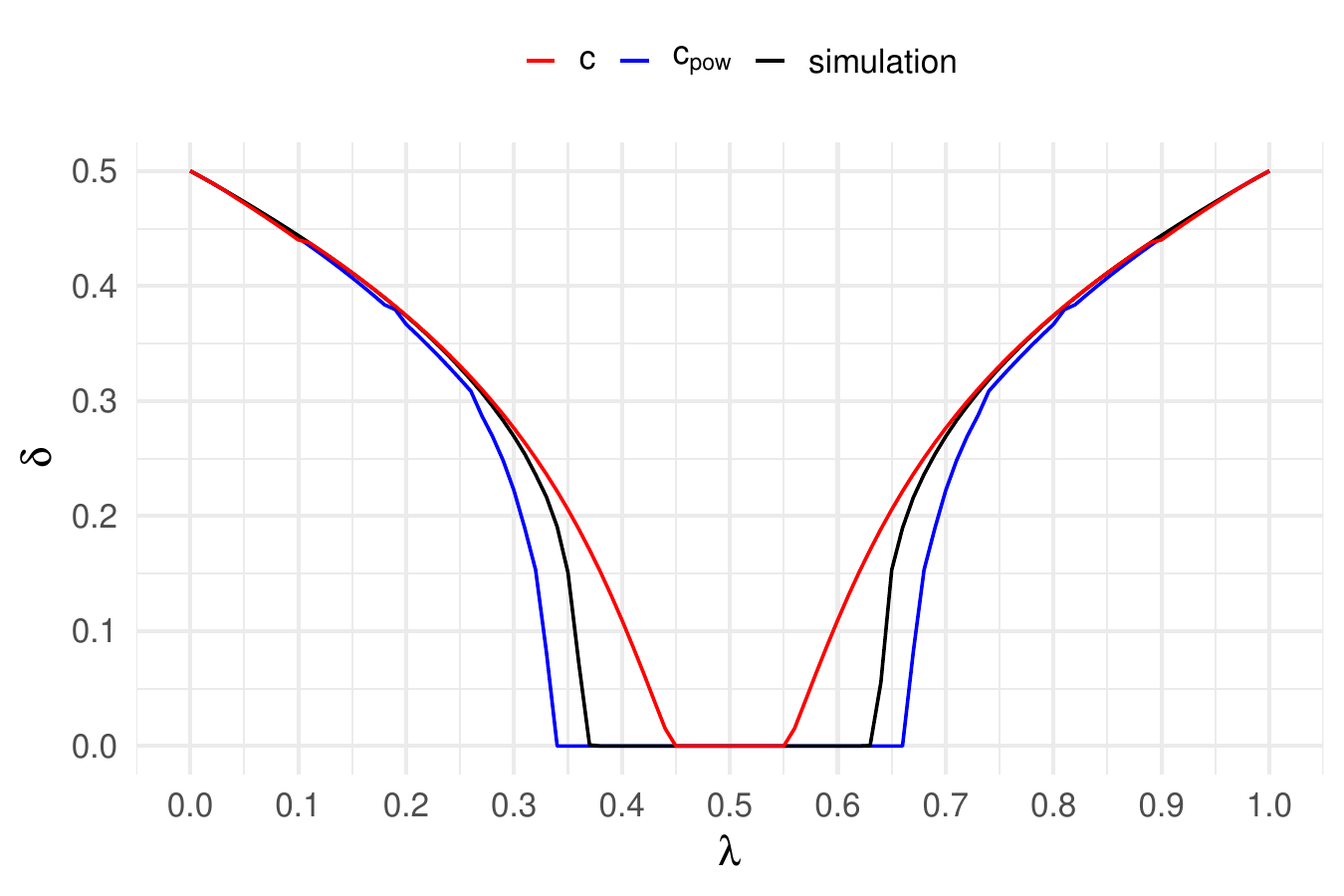}
\caption{Coefficients $c_{\textnormal{pow}}$ for the power-law model, $c$ for the exponential fit and the values of the Monte Carlo simulations }
\label{fig:power_exp}
\end{figure}
\begin{figure}
\centering
\includegraphics[width=0.45\textwidth]{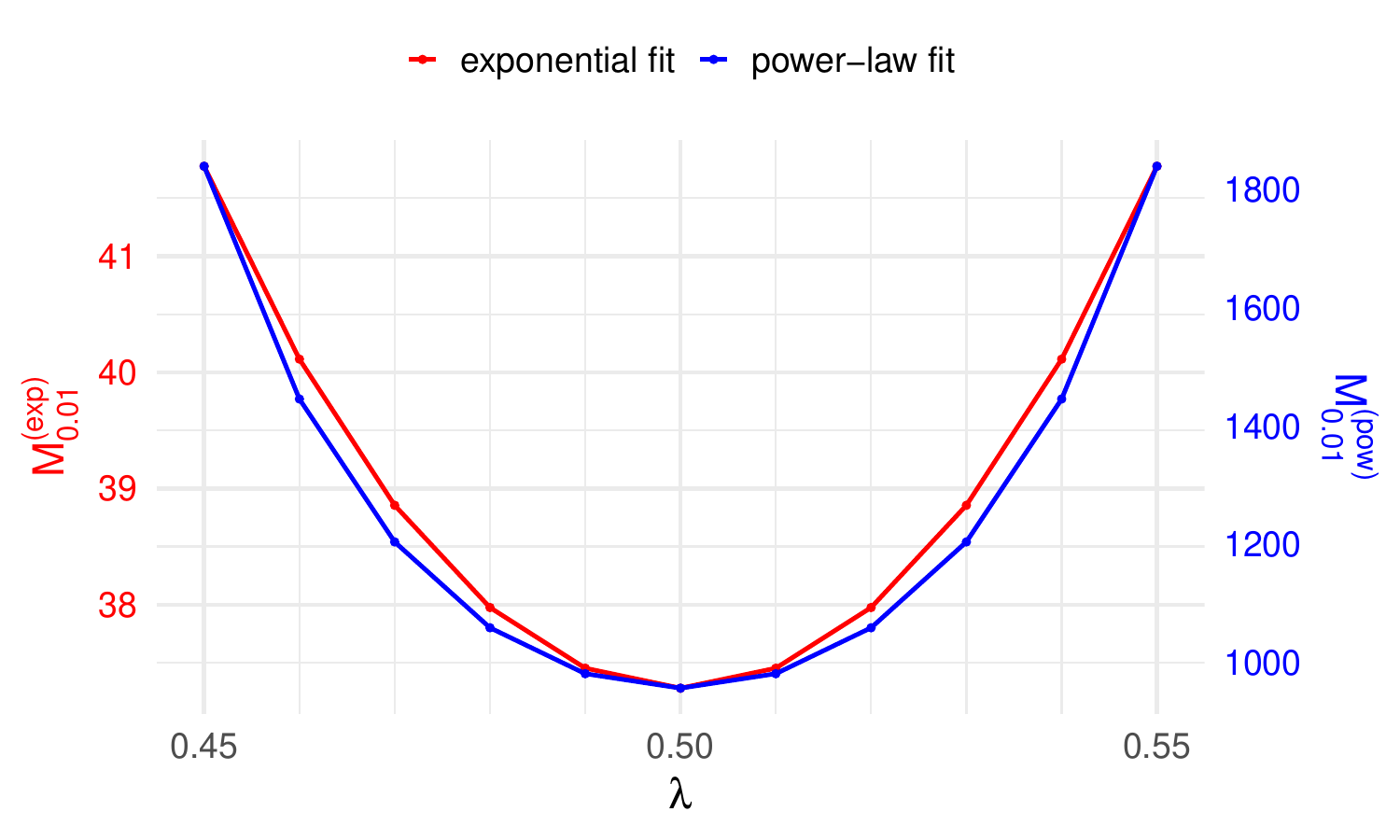}
\caption{Effective block size for getting a density as a function of $\lambda$, for the exponential and the power-law fits}
\label{fig:Meps_all}
\end{figure}

Figure~\ref{fig:Meps_all} shows, for each value of the mixing parameter
$\lambda$, the effective block size $M_{0.01}(\lambda)$ at which the
exponentially fitted BA predicts a density as small as
$\rho_{M}(\lambda)=0.01$.  The value of $M_{0.01}(\lambda)$ is obtained by
inverting the constrained exponential fit \eqref{eq:expfit_block}, 
namely
\[
   M_{0.01}(\lambda)
   = \frac{1}{b(\lambda)}\,
     \log\!\left(\frac{a(\lambda)}{0.01 - c(\lambda)}\right),
\]
for $c(\lambda)<0.01$.
This plot
provides a quantitative measure of how large a block must be before the
block-approximation predicts a practically vanishing density.  In particular,
near $\lambda=1/2$, the required block size is of the order of $M\approx 38$, much larger than our maximal block size $13$ in Figure~\ref{f:fig-dens2}. 

However, it is natural to ask whether describing the decay of the block density $\rho(M)$ as an exponential relaxation toward an asymptotic plateau is too restrictive to capture the behavior observed in the diploid dynamics. Hence, in order to include a power-law decay near the critical points, we complement the analysis based on a single-rate exponential
decay by considering a power-law parametrization of the block
density.  For each value of~$\lambda$ and each accessible block size
$M\in\{3,5,7,9,11,13\}$, we fit the functional form
\begin{equation}
\label{e:power}
\rho^{(2)}_M(\lambda)
  = c_{\mathrm{pow}}(\lambda) + a_{\mathrm{pow}}(\lambda)\, M^{-b_{\mathrm{pow}}(\lambda)},
\end{equation}
with $a_{\mathrm{pow}},\, b_{\mathrm{pow}},\, c_{\mathrm{pow}}\ge0$.
 The resulting asymptotic densities $c_\text{pow}$ are shown in
Figure~\ref{fig:power_exp}, together with the exponential fit and the  Monte Carlo stationary densities.  We observe that both extrapolation schemes predict the existence of a zero-density region, with the exponential fit \eqref{eq:expfit_block} underestimating and the power-law \eqref{e:power} overestimating its size.

%
%

Finally, we remark that the ability of the BA 
to detect transition points is restored if 
one considers a three \gls{eca} mixture (i.e., a \emph{triploid}) adding to the $60$ and 
$102$ rules the $0$ rule, as well. 
More precisely, one chooses the function $\phi$ as in \eqref{mod010}
with 
$f_1$ the zero rule, 
$f_2$ the rule 60, 
and 
$f_3$ the rule 102, 
with 
$\xi_1=\epsilon$,
$\xi_2=1-\lambda$,
and 
$\xi_3=\lambda-\epsilon$, 
so that for $\epsilon=0$ the original diploid is recovered.
Note that $0\leq\epsilon \leq \lambda \leq 1$.

The MF approximation, again, fails in identifying the 
transition, indeed 
the MF equation 
\begin{align*}
\delta
=
\varrho(1)
=
&
(\lambda-\epsilon)(1-\delta)\delta^2
+
(1-\epsilon)(1-\delta)\delta^2
\\
&
+
(1-\lambda)(1-\delta)^2\delta
+
(1-\lambda)(1-\delta)\delta^2
\\
&
+
(1-\epsilon)(1-\delta)^2\delta
+
(\lambda-\epsilon)(1-\delta)^2\delta
,
\end{align*}
has the solutions $\delta=0$
and
$\delta=(1-2\epsilon)/[2(1-\epsilon)]$ for every value of $\lambda$.
The second solution is positive for 
$\epsilon\in[0,1/2)$, equal to zero for $\epsilon=1/2$, and 
negative (that is to say, meaningless) for $\epsilon>1/2$. 
Thus, the MF approximation predicts, for every admissible $\lambda$,
both the zero density 
and a non-zero density solution for 
$\epsilon\in[0,1/2)$, while for $\epsilon\ge1/2$ the sole zero density measure is found. 
We can thus conclude that, given $\epsilon$, the MF approximation 
is not able to detect any transition. 

\begin{figure}
\centering
\includegraphics[width=0.48\textwidth,height=0.39\textwidth]{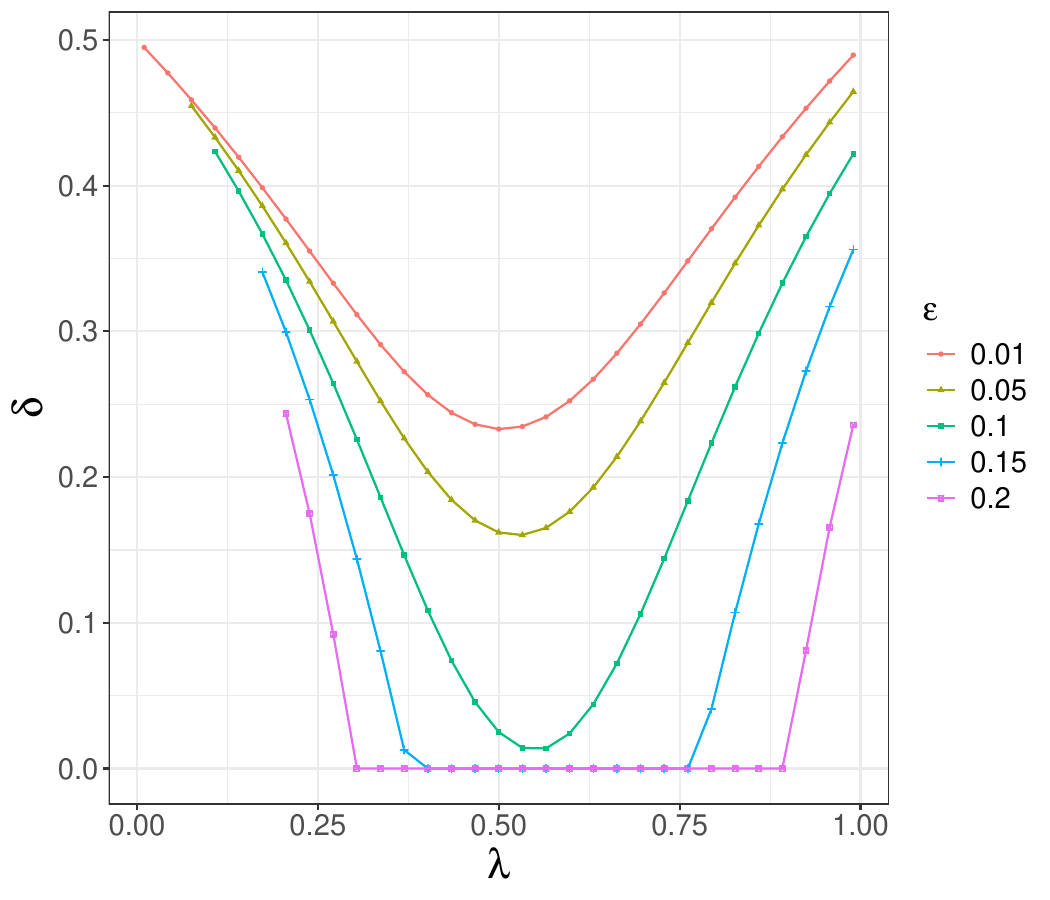}
\caption{Density for the 
three rule mixture as function
of $\lambda\in[\epsilon,1]$ 
computed via the
$5$-cell BA 
for $\epsilon=0.01$,
$\epsilon=0.05$,
$\epsilon=0.10$,
$\epsilon=0.15$,
$\epsilon=0.20$.
We report the solution found 
initializing the iterative equations with the $1/2$ Bernoulli measure.
The corresponding values predicted by the MF approximation are 
$0.495$, 
$0.474$, 
$0.444$, 
$0.412$, 
and
$0.375$.
}
\label{f:fig-dens3}
\end{figure}

As reported in Figure~\ref{f:fig-dens3}, the MF
prediction is largely improved by the BA.
If $\epsilon$ is large enough, two 
transitions between the zero and the non-zero density 
are detected.
Indeed, there exists an interval of values of $\lambda$ in which 
the zero density measure is found as the solution of the  iterative equations initialized with the Bernoulli measure, 
meaning that the associated fixed point in the measure space 
is stable. 

\section{Conclusions}
\label{s:future}

We have studied the dynamics of a probabilistic cellular automaton obtained as a
stochastic mixture of the additive ECA rules~60 and~102.
The model provides a nontrivial example of a probabilistic combination of two
mirror-symmetric linear automata, whose stochastic dynamics exhibit an absorbing
zero stationary state for a finite interval of values of the mixing parameter.
The physics interpretation would be that the competition between the right and left moving rules $60$ and $102$ leads to the stabilization of an intermediate zero-density phase.

A central result of this work is that standard block approximation schemes fail
systematically to reproduce this absorbing stationary state.
The failure can be traced to the additive and mirror symmetries of the underlying
deterministic rules, which impose global algebraic constraints on the evolution
of configurations.
Such constraints are not representable within a closure based on finitely many
local marginals, and the resulting block approximation equations preserve
spurious stationary solutions that do not correspond to invariant measures of
the exact stochastic dynamics.
This shows that, for probabilistic mixtures of linear cellular automata,
mean-field-type closures may yield qualitatively incorrect predictions even in
regimes where the exact finite system almost surely reaches an absorbing state.

The contrast with percolation-type PCA, such as the Domany--Kinzel model (\cite{domany1984}), is
instructive.
In those systems the update rule is monotone: activity can only be generated
locally from pre-existing activity, and once a configuration becomes locally
inactive it cannot regenerate activity at later times.
As a consequence, the long-time behavior is determined by local interactions,
and low-order block approximations are sufficient to capture the qualitative
structure of the extinction transition.

More generally, our results indicate that approximation schemes based on finite
local marginals may be unreliable for stochastic cellular automata whose dynamics
are governed by additive or other algebraic symmetries.
In such systems, global constraints on configuration space can dominate the
long-time behavior, leading to qualitative discrepancies between the exact
process and its finite-order closures.
Developing approximation methods capable of incorporating such constraints
remains an interesting direction for future work.

\begin{acknowledgments}
ENMC  
thanks GNFM and the PRIN 2022 project
``Mathematical Modelling of Heterogeneous Systems (MMHS)",
financed by the European Union - Next Generation EU,
CUP B53D23009360006, Project Code 2022MKB7MM, PNRR M4.C2.1.1.
\end{acknowledgments}



\bibliographystyle{unsrt}

\bibliography{cs-pca_misc}

\end{document}